%% file: main.tex
\documentclass[10pt,a4paper]{article}
\usepackage[T1]{fontenc}
\usepackage{graphicx}
\usepackage{color}

\input{macros}
\usepackage{a4wide}
\usepackage{mathpartir}
\usepackage{float}
\usepackage{tabularx}
\usepackage[inference]{semantic}
\usepackage{esvect}
\usepackage{quantikz}
\usepackage{multicol}
\usepackage{authblk}

\usepackage{amsthm}
\newtheorem{definition}{Definition}
\newtheorem{theorem}{Theorem}
\newtheorem{lemma}{Lemma}
\begin{document}
\title{Circuit Width Estimation via Effect Typing\\ and Linear Dependency (Long Version)}
%
%
\author[1,2]{Andrea Colledan}
\author[1,2]{Ugo Dal Lago}
\affil[1]{University of Bologna, Italy}
\affil[2]{INRIA Sophia Antipolis, France}
\date{}
\maketitle              
\begin{abstract}
Circuit description languages are a class of quantum programming languages in which programs are classical and produce a \textit{description} of a quantum computation, in the form of a \textit{quantum circuit}. Since these programs can leverage all the expressive power of high-level classical languages, circuit description languages have been successfully used to describe complex and practical quantum algorithms, whose circuits, however, may involve many more qubits and gate applications than current quantum architectures can actually muster. In this paper, we present \PQR, a circuit description language endowed with a linear dependent type-and-effect system capable of deriving parametric upper bounds on the width of the circuits produced by a program. We prove both the standard type safety results and that the resulting resource analysis is correct with respect to a big-step operational semantics. We also show that our approach is expressive enough to verify realistic quantum algorithms.

\end{abstract}

\section{Introduction}

With the promise of providing efficient algorithmic 
solutions to many problems \cite{grover,hhl,qft}, some of which are traditionally believed to be intractable \cite{shor}, quantum computing is the subject of intense investigation by various research communities within computer science, not least that of programming language theory \cite{quantum-languages,survey-gay,survey-selinger}. Various proposals for idioms capable of tapping into this new computing paradigm have appeared in the literature since the late 1990s. Some of these approaches turn out to be fundamentally new \cite{qpl,qml,qgcl}, while many others are strongly inspired by classical languages and traditional programming paradigms \cite{qcl,quantum-lambda-calculus,proto-quipper-s,qwire}.

One of the major obstacles to the practical adoption of quantum algorithmic 
solutions is the fact that despite huge efforts by scientists and 
engineers alike, it seems that reliable quantum hardware, contrary to classical one, does not scale too easily: although quantum architectures with up to a couple hundred qubits have recently seen the light \cite{osprey,sycamore,jiuzhang}, it is not yet clear whether the so-called quantum advantage \cite{quantum-supremacy} is a concrete possibility, given the tremendous challenges posed by the quantum decoherence problem \cite{quantum-decoherence}.

This entails that software which makes use of quantum hardware 
must be designed with great care: whenever part of a computation has to be run on quantum hardware, the amount of resources it 
needs, and in particular the amount of qubits it uses, should be kept to a 
minimum. More generally, a fine control over the low-level aspects of the 
computation, something that we willingly abstract from in most cases when dealing with classical computations, should be exposed to the programmer in the quantum case.
This, in turn, has led to the development and adoption of many domain-specific 
programming languages and libraries in which the programmer \textit{explicitly} manipulates qubits and quantum circuits, while still making use of all the features of a high-level classical programming language. This is the case of the \Qiskit\ and \Cirq\ libraries \cite{quantum-software}, but also of the \Quipper\ language \cite{quipper-intro,quipper}.

At the fundamental level, \Quipper\ is a circuit description language embedded in \Haskell. Because of this, \Quipper\ inherits all the expressiveness of the high level, higher-order functional programming language that is its host, but for the same reason it also lacks a formal semantics. Nonetheless, over the past few years, a number of calculi, collectively known as the \PQ\ language family, have been developed to formalize interesting fragments and extensions of \Quipper\ and its type system \cite{proto-quipper-s,proto-quipper-m}. Extensions include, among others, dynamic lifting \cite{proto-quipper-l,proto-quipper-dyn,proto-quipper-k} and dependent types \cite{proto-quipper-d,proto-quipper-d-intro}, but resource analysis is still a rather unexplored research direction in the \PQ\ community \cite{gate-count}.

The goal of this work is to show that type systems indeed enable the 
possibility of reasoning about the size of the circuits produced by a 
\PQ\ program. Specifically, we show how linear dependent types in the 
form given by Gaboardi and Dal Lago \cite{linear-dependent-types-relative-completeness,linear-dependent-types-privacy,linear-dependent-types-cbv,geometry-of-types} can be adapted to 
\PQ, allowing to derive upper bounds on circuit widths that 
are parametric on the size of the problem. This enables a form of static analysis of the resource consumption of circuit families and, consequently, of the quantum algorithms described in the language. Technically, a key ingredient of this analysis, besides linear dependency, is a novel form of effect typing in which the quantitative information coming from linear dependency informs the effect system and allows it to keep circuit widths under control.

The rest of the paper is organized as follows. Section \ref{sect:overview} informally explores the problem of estimating the width of circuits produced by \Quipper, while also introducing the language.
Section \ref{sect:protoquipper} provides a more formal definition of the \PQ\ language. In particular, it gives an overview of the system of simple types due to Selinger and Rios \cite{proto-quipper-m}, which however is not designed to reason about the size of circuits.
We then move on to the most important technical contribution of this work, namely the linear dependent and effectful type system, which is introduced in Section \ref{sect:typesystem} and proven to guarantee both type safety and a form of total correctness in Section \ref{sect:safetycorrectness}.
Section \ref{sec:practical-examples} is dedicated to an example of a practical application of our type and effect system, that is, a program that builds the Quantum Fourier Transform (QFT) circuit \cite{qft,nielsen-chuang} and which is verified to do so without any ancillary qubits.

\section{An Overview on Circuit Width Estimation}\label{sect:overview}
\input{overview}

\section{The \textsf{Proto-Quipper Language}}\label{sect:protoquipper}
\input{the-proto-quipper-language}

\section{Incepting Linear Dependency and Effect Typing}\label{sect:typesystem}
\input{incepting-linear-dependency-and-effect-typing}

\section{Type Safety and Correctness}\label{sect:safetycorrectness}
\input{type-safety-and-correctness}

\section{A Practical Example}
\input{practical-examples}

\section{Related Work}
\input{related}

\section{Generalization to Other Resource Types}
\input{generalization}

\section{Conclusion and Future Work}
\input{conclusion}
\

\bibliographystyle{plain}
\bibliography{bibliography}

\end{document}

%% file: macros.tex
\usepackage{bbm}
\usepackage{mathrsfs}
\usepackage{amsmath}
\usepackage{amssymb}
\usepackage{stmaryrd}

\makeatletter
\newcommand{\xmultimap}[2][]{\ext@arrow 0359\multimapfill@{#1}{#2}}
\newcommand*{\multimapfill@}{%
	\arrowfill@\relbar\relbar\multimap}
\makeatother

\newcommand{\gateset}{\mathscr{G}}
\newcommand{\gatesetoftype}[2]{\mathscr{G}(#1,#2)}
\newcommand{\circset}{\mathit{CIRC}}
\newcommand{\mvalset}{\mathit{BVAL}}
\newcommand{\mtypeset}{\mathit{BTYPE}}
\newcommand{\labOne}{\ell} \newcommand{\labTwo}{k} \newcommand{\labThree}{t} \newcommand{\labFour}{q}
\newcommand{\wtypeOne}{w}
\newcommand{\lcOne}{Q} \newcommand{\lcTwo}{L} \newcommand{\lcThree}{H} \newcommand{\lcFour}{K}
\newcommand{\gateOne}{g} 
\newcommand{\circuitOne}{\mathcal{C}} \newcommand{\circuitTwo}{\mathcal{D}} \newcommand{\circuitThree}{\mathcal{E}}
\newcommand{\mtypeOne}{T} \newcommand{\mtypeTwo}{U}
\newcommand{\emptylc}{\emptyset}
\newcommand{\cidentity}[1]{id_{#1}}
\newcommand{\gateapp}[3]{#1(#2) \to #3}
\newcommand{\unitt}{\mathbbm{1}}
\newcommand{\qubitt}{\mathsf{Qubit}}
\newcommand{\qubitts}{\mathsf{Q}}
\newcommand{\bitt}{\mathsf{Bit}}
\newcommand{\natt}{\mathsf{Nat}}
\newcommand{\circjudgment}[3]{#1 : #2 \to #3}
\newcommand{\mjudgment}[3]{#1 \vdash_w #2 : #3}
\newcommand{\struct}[1]{\bar{#1}}
\newcommand{\concat}[2]{#1::#2}
\newcommand{\width}[1]{\operatorname{width}(#1)}
\newcommand{\discarded}[1]{\operatorname{discarded}(#1)}
\newcommand{\inits}[1]{\operatorname{new}(#1)}
\newcommand{\proves}{\triangleright}

\newcommand{\vsubSet}{\mathit{VSUB}}
\newcommand{\isubSet}{\mathit{ISUB}}
\newcommand{\termSet}{\mathit{TERM}}
\newcommand{\indexSet}{\mathit{INDEX}}
\newcommand{\valSet}{\mathit{VAL}}
\newcommand{\typeSet}{\mathit{TYPE}}
\newcommand{\ptypeSet}{\mathit{PTYPE}}
\newcommand{\varOne}{x} \newcommand{\varTwo}{y} \newcommand{\varThree}{z}
\newcommand{\ivarOne}{i} \newcommand{\ivarTwo}{j} \newcommand{\ivarThree}{e}
\newcommand{\contextOne}{\Gamma}
\newcommand{\pcontextOne}{\Phi}
\newcommand{\icontextOne}{\Theta} 
\newcommand{\termOne}{M} \newcommand{\termTwo}{N} 
\newcommand{\indexOne}{I} \newcommand{\indexTwo}{J} \newcommand{\indexThree}{E}
\newcommand{\valOne}{V} \newcommand{\valTwo}{W} \newcommand{\valThree}{X} \newcommand{\valFour}{Y} \newcommand{\valFive}{Z}
\newcommand{\typeOne}{A} \newcommand{\typeTwo}{B} \newcommand{\typeThree}{C}
\newcommand{\ptypeOne}{P} \newcommand{\ptypeTwo}{R}
 
\newcommand{\emptycontext}{\emptyset}

\newcommand{\emptysub}{\emptyset}
\newcommand{\unitv}{*}
\newcommand{\tuple}[2]{\langle#1,#2\rangle}
\newcommand{\abs}[3]{\lambda #1_{#2}.#3}
\newcommand{\app}[2]{#1\,#2}
\newcommand{\liftoperator}{\operatorname{\mathsf{lift}}}
\newcommand{\lift}[1]{\liftoperator #1}
\newcommand{\force}[1]{\operatorname{\mathsf{force}} #1}
\newcommand{\boxoperator}{\operatorname{\mathsf{box}}}
\newcommand{\boxt}[2]{\boxoperator_{#1} #2}
\newcommand{\boxedCirc}[3]{(#1,#2,#3)}
\newcommand{\applyoperator}{\operatorname{\mathsf{apply}}}
\newcommand{\apply}[2]{\applyoperator(#1,#2)}
\newcommand{\letoperator}{\mathsf{let}}
\newcommand{\letin}[3]{\letoperator\; #1 = #2 \;\mathsf{in}\; #3}
\newcommand{\dest}[4]{\letoperator\; \tuple{#1}{#2} = #3 \;\mathsf{in}\; #4}
\newcommand{\returnoperator}{\mathsf{return}}
\newcommand{\return}[1]{\returnoperator\; #1}
\newcommand{\foldoperator}{\operatorname{\mathsf{fold}}}
\newcommand{\fold}[3]{\foldoperator_{#1}\, {#2}\; {#3}}
\newcommand{\nil}{\mathsf{nil}}
\newcommand{\cons}[2]{\mathsf{cons}\; {#1}\; {#2}}
\newcommand{\lineararrow}{\multimap}
\newcommand{\arrowt}[4]{#1 \multimap_{#3,#4} #2}
\newcommand{\bang}[1]{\operatorname{!}#1}
\newcommand{\circt}[3]{\operatorname{\mathsf{Circ}}^{#1}(#2,#3)}
\newcommand{\tensor}[2]{#1 \otimes #2}
\newcommand{\listt}[2]{\mathsf{List}^{#1}\,{#2}}
\newcommand{\natOne}{n}  
\newcommand{\iplus}[2]{#1 + #2}
\newcommand{\iminus}[2]{#1 - #2}
\newcommand{\imult}[2]{#1 \times #2}
\newcommand{\imax}[2]{\mathsf{max}(#1,#2)}
\newcommand{\imaximum}[3]{\mathsf{max}_{#1 < #2}\,#3}
\newcommand{\cjudgment}[6]{#1;#2;#3 \vdash_c #4 : #5 ; #6}
\newcommand{\vjudgment}[5]{#1;#2;#3 \vdash_v #4 : #5}
\newcommand{\wfjudgment}[2]{#1 \vdash #2}
\newcommand{\subtypejudgment}[3]{#1 \vdash_s #2 <: #3}
\newcommand{\sametypejudgment}[3]{#1 \vdash_s #2 <:> #3}
\newcommand{\leqjudgment}[3]{#1 \vDash #2 \leq #3}
\newcommand{\eqjudgment}[3]{#1 \vDash #2 = #3}
\newcommand{\interpret}[1]{\llbracket #1 \rrbracket}
\newcommand{\config}[2]{(#1,#2)}
\newcommand{\eval}{\Downarrow}
\newcommand{\freshlabelsfunction}{\operatorname{freshlabels}}
\newcommand{\freshlabels}[1]{\freshlabelsfunction(#1)}
\newcommand{\appendfunction}{\operatorname{append}}
\newcommand{\append}[3]{\appendfunction(#1,#2,#3)}
\newcommand{\sub}[2]{[#1/#2]}
\newcommand{\isub}[2]{\{#1/#2\}}
\newcommand{\cconfigjudgment}[5]{#1 \vdash #2 : #3 ; #4; #5}
\newcommand{\vconfigjudgment}[4]{#1 \vdash #2 : #3 ; #4}
\newcommand{\realizes}[3]{#1 \Vdash_{#2} #3}
\newcommand{\reducible}[4]{#1 \Vdash_{#2}^{#3} #4}
\newcommand{\vsubOne}{\gamma}
\newcommand{\isubOne}{\theta}

\newcommand{\vimplements}[3]{#1 \vDash_{#2} #3}
\newcommand{\iimplements}[2]{#1 \vDash #2}

\newcommand{\rcount}[1]{\#(#1)}

\newcommand{\maxf}[2]{\operatorname{max}(#1,#2)}

\newcommand{\forceindent}{\phantom{m}}

\newcommand{\PQ}{\textsf{Proto-Quipper}}
\newcommand{\PQM}{\textsf{Proto-Quipper-M}}

\newcommand{\PQS}{\textsf{Proto-Quipper-S}}
\newcommand{\PQD}{\textsf{Proto-Quipper-D}}

\newcommand{\PQR}{\textsf{Proto-Quipper-R}}
\newcommand{\Haskell}{\texttt{Haskell}}
\newcommand{\LiquidHaskell}{\texttt{Liquid Haskell}}
\newcommand{\Quipper}{\texttt{Quipper}}
\newcommand{\Qiskit}{\texttt{Qiskit}}
\newcommand{\Cirq}{\texttt{Cirq}}

\newcommand{\crl}{\textsf{CRL}}


%% file: overview.tex
\label{sec:overview}

\Quipper\ allows programmers to describe quantum circuits in a high-level and elegant way, using both gate-by-gate and circuit transformation approaches. \Quipper\ also supports  hierarchical and parametric circuits, thus promoting a view in which circuits become first-class citizens. \Quipper\ has been shown to be scalable, in the sense that it has been effectively used to describe complex quantum algorithms that easily translate to circuits involving trillions of gates applied to millions of qubits. The language allows the programmer to optimize the circuit, e.g. by using ancilla qubits for the sake of reducing the circuit depth, or recycling qubits that are no longer needed. 

One feature that \Quipper\ lacks is a methodology for \emph{statically} proving that important parameters --- such as the the width --- of the 
underlying circuit are below certain limits, which of course can be parametric on the input size of the circuit. If this kind of analysis were available, then it would be possible to derive bounds on the number of qubits needed to solve any instance of a problem, and ultimately to know in advance how big of an instance can be \textit{possibly} solved given a fixed amount of qubits. 

In order to illustrate the kind of scenario we are reasoning about, this section offers some simple examples of \Quipper\ programs, showing in what sense we can think of capturing the quantitative information that we are interested in through types and effect systems and linear dependency. We proceed at a very high level for now, without any ambition of formality.

Let us start with the example of Figure \ref{fig:dumbnot}. The \Quipper\ function on the left builds the quantum circuit on the right: an (admittedly  contrived) implementation of the quantum not operation. The \texttt{dumbNot} function implements negation using a controlled not gate and an ancillary qubit \texttt{a}, which is initialized and discarded within the body of the function. This qubit does not appear in the interface of the circuit, but it clearly adds to its overall width, which is $2$.

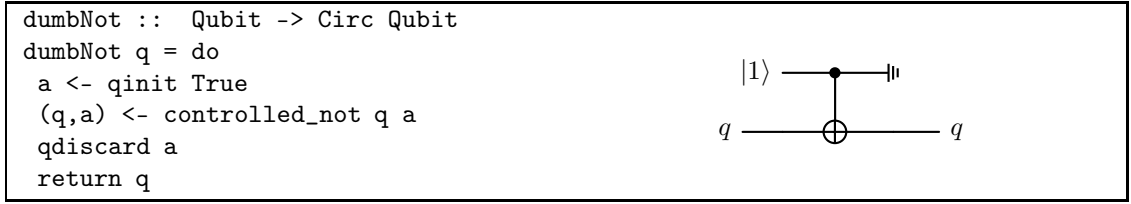
\begin{figure}[h]
	\centering
	\fbox{
	\begin{minipage}{0.475\textwidth}
		\texttt{\noindent
		dumbNot :: Qubit -> Circ Qubit\\
		dumbNot q = do\\
		\forceindent a <- qinit True\\
		\forceindent (q,a) <- controlled\_not q a\\
		\forceindent qdiscard a\\
		\forceindent return q
		}
	\end{minipage}
	\begin{minipage}{0.475\textwidth}
		\centering
		\begin{quantikz}[column sep=1.5em]
			\setwiretype{n}& \lstick{$\ket{1}$} & \ctrl{1} \setwiretype{q} & \ground{} \\
			\lstick{$q$} && \targ & && \rstick{$q$}
		\end{quantikz}
	\end{minipage}
	}
\caption{An implementation of the quantum not operation using an ancilla.}
\label{fig:dumbnot}
\end{figure}

Consider now the higher-order function in Figure \ref{fig:iter}. This function takes as input a circuit building function \texttt{f}, an integer \texttt{n} and describes the circuit obtained by applying \texttt{f}'s circuit \texttt{n} times to the input qubit \texttt{q}.
It is easy to see that the width of the circuit produced in 
output by \texttt{iter dumbNot n} is equal to $2$, even though, overall, the number of qubits initialized during the computation is equal to $n$. The point is that each ancilla is created only \emph{after} the previous one has been discarded, thus enabling a form of qubit recycling.

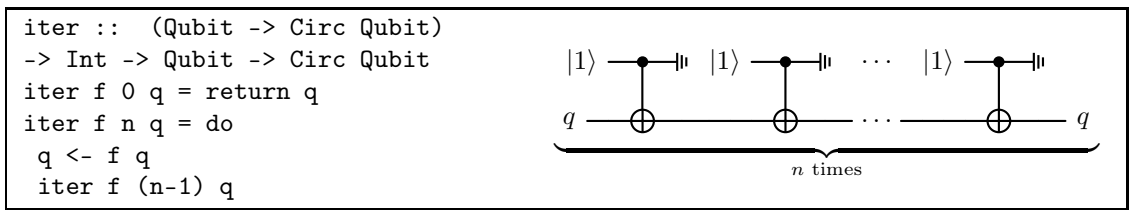
\begin{figure}[h]
	\centering
	\fbox{
		\begin{minipage}{0.45\linewidth}
			\texttt{\noindent
				iter :: (Qubit -> Circ Qubit)\\ -> Int -> Qubit -> Circ Qubit\\
				iter f 0 q = return q\\
				iter f n q = do\\
				\forceindent q <- f q\\
				\forceindent iter f (n-1) q
			}
		\end{minipage}
		\begin{minipage}{0.5\linewidth}
			\centering
			$\underbrace{
			\begin{quantikz}[column sep=0.8em]
				\setwiretype{n}& \lstick{$\ket{1}$} & \ctrl{1} \setwiretype{q} & \ground{}  & \setwiretype{n} && \lstick{$\ket{1}$} & \ctrl{1} \setwiretype{q} & \ground{} & \setwiretype{n} \ \ldots \ &&&\lstick{$\ket{1}$} & \ctrl{1} \setwiretype{q} & \ground{} \\
				\lstick{$q$} && \targ{} & &&& &\targ{}&& \ \ldots \ &&&& \targ{}&& \rstick{$q$}
				\end{quantikz}
			}_{n \text{ times}}$
		\end{minipage}
	}
	\caption{A higher-order function which iterates a circuit-building function \texttt{f} on an input qubit \texttt{q} and the result of its application to the \texttt{dumbNot} function from Figure \ref{fig:dumbnot}.}
	\label{fig:iter}
\end{figure}

Is it possible to statically analyze the width of the circuit produced in output by \texttt{iter dumbNot n} so as to conclude that it is constant and equal to $2$? What techniques can we use? Certainly, the presence of higher-order types complicates the problem, already in itself non-trivial. The approach we propose in this paper is based on two ingredients. The first is the so-called effect typing \cite{types-and-effects}. In this context, the effect produced by the program is nothing more than the circuit and therefore it is natural to think of an effect system in which the width of such circuit, and only that, is exposed. Therefore, the arrow type $\typeOne \to \typeTwo$ should be decorated with an expression indicating the width of the circuit produced by the corresponding function when applied to an argument of type $\typeOne$. Of course, the width of an individual circuit is a natural number, so it would make sense to annotate the arrow with a natural number. For technical reasons, however, it will also be necessary to keep track of another natural number, corresponding to the amount of wire resources that the function captures from the surrounding environment. This necessity stems from a need to keep close track of wires even in the presence of data hiding, and will be explained in further detail in Section \ref{sec:incepting-linear-dependency-and-effect-typing}.

Under these premises, the \texttt{dumbNot} function would receive type $\qubitt \to_{2,0} \qubitt$, meaning that it takes as input a qubit and produces a circuit of width $2$ which outputs a qubit.
Note that the second annotation is $0$, since we do not capture anything from the function's environment, let alone a wire. Consequently, because \texttt{iter} iterates in sequence and because the ancillary qubit in \texttt{dumbNot} can be reused, the type of \texttt{iter dumbNot n} would also be $\qubitt \to_{2,0} \qubitt$.

\begin{figure}[h]
	\centering
	\fbox{
		\begin{minipage}{0.49\linewidth}
			\texttt{\noindent
				hadamardN :: [Qubit] -> Circ [Qubit]\\
				hadamardN [] = return []\\
				hadamardN (q:qs)  = do\\
				\forceindent q <- hadamard  q\\
				\forceindent qs <- hadamardN qs\\
				\forceindent return (q:qs)
			}
		\end{minipage}
		\begin{minipage}{0.46\linewidth}
			\centering
			\begin{quantikz}[column sep=2em, row sep=1em]
				\lstick{$q$} & \gate{H} & \rstick{$q$}\\
				\lstick{$\mathit{qs}_1$} & \gate{H} & \rstick{$\mathit{qs}_1$}\\
				\wave&&\\
				\lstick{$\mathit{qs}_n$} & \gate{H} & \rstick{$\mathit{qs}_n$}
			\end{quantikz}
		\end{minipage}
	}
	\caption{The \texttt{hadamardN} function implements a circuit family where circuits have width linear in their input size.}
	\label{fig:hadamardn}
\end{figure}
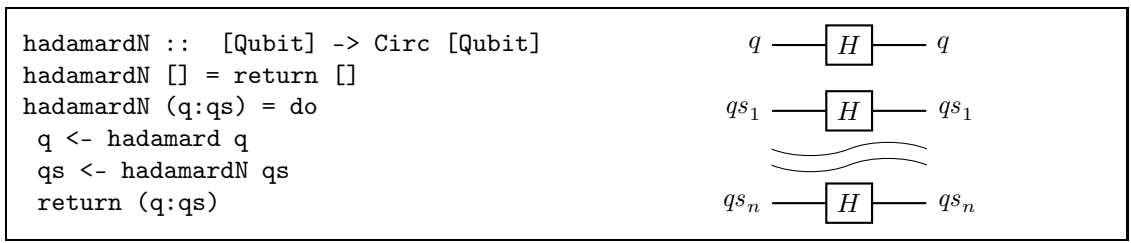

Let us now consider a slightly different situation, in which the width of the produced circuit is not constant, but rather increases proportionally to the circuit's input size. Figure \ref{fig:hadamardn} shows a \Quipper\ function that returns a circuit on $n$ qubits in which the Hadamard gate is applied to each qubit. This simple circuit represents the preprocessing phase of many quantum algorithms, including Deutsch-Josza \cite{deutsch-jozsa} and Grover's \cite{grover}. It is obvious that this function works on inputs of arbitrary size, and therefore we can interpret it as a circuit family, parametrized on the length of the input list of qubits. This quantity, although certainly a natural number, is unknown statically and corresponds precisely to the width of the produced circuit. A question therefore arises as to whether the kind of effect typing we briefly hinted at in the previous paragraph is capable of dealing with such a function. Certainly, the expressions used to annotate arrows cannot be, like in the previous case, mere \emph{constants}, as they clearly depend on the size of the input list. Is there a way to reflect this dependency in types? Certainly, one could go towards a fully-fledged notion of dependent types, like the ones proposed in \cite{proto-quipper-d}, but a simpler approach, in the style of Dal Lago and Gaboardi's linear dependent types \cite{linear-dependent-types-relative-completeness,linear-dependent-types-privacy,linear-dependent-types-cbv,geometry-of-types} turns out to be enough for this purpose. This is precisely the route that we follow in this paper. In this approach, terms can indeed appear in types, but that is only true for a very restricted class of terms, disjoint from the ordinary ones, called \emph{index terms}. As an example, the type of the function \texttt{hadamardN} above could become $\listt{\ivarOne}{\qubitt} \to_{\ivarOne,0} \listt{\ivarOne}{\qubitt}$, where $\ivarOne$ is an \textit{index variable}. The meaning of the type would thus be that \texttt{hadamardN} takes as input any list of qubits of length $\ivarOne$ and produces a circuit of width at most $\ivarOne$ which outputs $\ivarOne$ qubits. The language of indices is better explained in Section \ref{sec:incepting-linear-dependency-and-effect-typing}, but in general we can say that indices are arithmetical expressions over natural numbers and index variables, and can thus express non-trivial dependencies between input sizes and corresponding circuit widths.

%% file: the-proto-quipper-language.tex
\label{sec:the-proto-quipper-language}

This section aims at introducing the \PQ\ family of calculi to the non-specialist, without any form of resource analysis. At its core, \PQ\ is a linear lambda calculus with bespoke constructs to build and manipulate circuits. Circuits are built as the side-effect of a computation, behind the scenes, but they can also appear and be manipulated as data in the language.

\begin{figure}[h]
	\centering
	\fbox{\parbox{.98\textwidth}{\centering
		\begin{tabular}{l l r l}
			Types	& $\typeSet$	& $\typeOne,\typeTwo$ & $
			::= \unitt
			\mid \wtypeOne
			\mid \bang{\typeOne}
			\mid \tensor{\typeOne}{\typeTwo}
			\mid \typeOne \lineararrow \typeTwo
			\mid \listt{}{\typeOne}
			\mid \circt{}{\mtypeOne}{\mtypeTwo}$\\
			Parameter types		& $\ptypeSet$	& $\ptypeOne,\ptypeTwo$ & $
			::= \unitt
			\mid \bang{\typeOne}
			\mid \tensor{\ptypeOne}{\ptypeTwo}
			\mid \listt{}{\ptypeOne}
			\mid \circt{}{\mtypeOne}{\mtypeTwo}$\\
			Bundle types	& $\mtypeset$	& $\mtypeOne,\mtypeTwo$ & $
			::= \unitt
			\mid \wtypeOne
			\mid \tensor{\mtypeOne}{\mtypeTwo}$
		\end{tabular}
	}}
	\caption{\PQ\ types.}
	\label{fig:pq-types}
\end{figure}

The types of \PQ\ are given in Figure \ref{fig:pq-types}. Speaking at a high level, we can say that \PQ\ types are generally linear. In particular, $\wtypeOne\in\{\bitt,\qubitt\}$ is a \textit{wire type} and is linear, while $\lineararrow$ is the linear arrow constructor. A subset of types, called \textit{parameter types}, represent the values of the language that are \textit{not} linear and that can therefore be copied. Any term of type $\typeOne$ can be \textit{lifted} into a duplicable parameter of type $\bang{\typeOne}$ if its type derivation does not require the use of linear resources.

\begin{figure}[h]
	\centering
	\fbox{\parbox{.98\textwidth}{\centering
	\begin{tabular}{l l r l}
		Terms	& $\termSet$	& $\termOne,\termTwo$ & $ ::= \app{\valOne}{\valTwo}
		\mid \dest{\varOne}{\varTwo}{\valOne}{\termOne}
		\mid \force{\valOne}
		\mid \boxt{\mtypeOne}{\valOne}$\\&&&$
		\mid \apply{\valOne}{\valTwo}
		\mid \return{\valOne}
		\mid \letin{\varOne}{\termOne}{\termTwo}$\\
		Values	& $\valSet$	& $\valOne,\valTwo$ & $
		::= \unitv
		\mid \varOne
		\mid \labOne
		\mid \abs{\varOne}{\typeOne}{\termOne}
		\mid \lift{\termOne}
		\mid \boxedCirc{\struct\labOne}{\circuitOne}{\struct{\labTwo}}
		\mid \tuple{\valOne}{\valTwo}$\\&&&$
		\mid \nil
		\mid \cons{\valOne}{\valTwo}
		\mid \fold{}{\valOne}{\valTwo}$\\
		Wire bundles	& $\mvalset$	& $\struct\labOne,\struct\labTwo$ & $
		::= \unitv
		\mid \labOne
		\mid \tuple{\struct\labOne}{\struct\labTwo}$
	\end{tabular}	
}}
	\caption{\PQ\ syntax.}
	\label{fig:pq-syntax}
\end{figure}

Now, let us informally dissect the language as presented in Figure \ref{fig:pq-syntax}, starting with the language of values. The main constructs of interest are  \textit{labels} and \textit{boxed circuits}.
A label $\labOne$ represents a reference to a free wire of the underlying circuit being built and is attributed a wire type $\wtypeOne\in\{\bitt,\qubitt\}$. Labels have to be treated linearly due to the no-cloning property of quantum states \cite{nielsen-chuang}. Arbitrary structures of labels form a subset of values which we call \textit{wire bundles} and which are given \textit{bundle types}.
On the other hand, a boxed circuit $\boxedCirc{\struct\labOne}{\circuitOne}{\struct\labTwo}$ represents a circuit object $\circuitOne$ as a datum within the language, together with its input and output interfaces, given as wire bundles $\struct\labOne$ and $\struct\labTwo$. Such a value is given type $\circt{}{\mtypeOne}{\mtypeTwo}$, where bundle types $\mtypeOne$ and $\mtypeTwo$ are the input and output types of the circuit, respectively. Boxed circuits can be copied, manipulated by primitive functions and, more importantly, applied to the underlying circuit.
This last operation, which lies at the core of \PQ's circuit-building capabilities, is possible thanks to the $\applyoperator$ operator. This operator takes as first argument a boxed circuit $\boxedCirc{\struct\labOne}{\circuitOne}{\struct\labTwo}$ and appends $\circuitOne$ to the underlying circuit $\circuitTwo$. How does $\applyoperator$ know \textit{where} exactly in $\circuitTwo$ to apply $\circuitOne$? Thanks to a second argument: a bundle of wires $\struct\labThree$ coming from the free output wires of $\circuitTwo$, which identifies the exact location where $\circuitOne$ is supposed to be appended.

The language is expected to be endowed with constant boxed circuits corresponding to fundamental gates (e.g. Hadamard, controlled not, etc.), but the programmer can also introduce their own boxed circuits via the $\boxoperator$ operator. Intuitively, $\boxoperator$ takes as input a circuit-building function and evaluates it in a sandboxed environment, on dummy arguments, in a way that leaves the underlying circuit unchanged. This evaluation produces a standalone circuit $\circuitOne$, which is then returned by the $\boxoperator$ operator as a boxed circuit $\boxedCirc{\struct\labOne}{\circuitOne}{\struct\labTwo}$.

Figure \ref{fig:pq-dumbnot} shows the \PQ\ term corresponding to the \Quipper\ program in Figure \ref{fig:dumbnot}, as an example of the use of the language. Note that $\letin{\tuple{\varOne}{\varTwo}}{\termOne}{\termTwo}$ is syntactic sugar for $\letin{\varThree}{\termOne}{\dest{\varOne}{\varTwo}{\varThree}{\termTwo}}$.  The $\mathit{dumbNot}$ function is given type $\qubitt \lineararrow \qubitt$ and builds the circuit shown in Figure \ref{fig:dumbnot} when applied to an argument.

\begin{figure}[h]
	\centering
	\fbox{\parbox{.98\textwidth}{\centering
	\begin{align*}
		\mathit{dumbNot} \triangleq \abs{q}{\qubitt}{\,&
			\letin{a}{\apply{\mathsf{INIT_1}}{\unitv}}{\\&
				\letin{\tuple{q}{a}}{\apply{\mathsf{CNOT}}{\tuple{q}{a}}}{\\&
					\letin{\_}{\apply{\mathsf{DISCARD}}{a}}{\\&
						\return{q}}}}}
\end{align*}}}
\caption{An example \PQ\ program. $\mathsf{INIT_1,CNOT}$ and $\mathsf{DISCARD}$ are primitive boxed circuits implementing the corresponding elementary operations.}
\label{fig:pq-dumbnot}
\end{figure}

On the classical side of things, it is worth mentioning that \PQ\ as presented in this section does \textit{not} support general recursion. A limited form of recursion on lists is instead provided via a primitive $\foldoperator$ constructor, which takes as argument a (copiable) step function of type $\bang{((\tensor{\typeTwo}{\typeOne}) \lineararrow \typeTwo)}$, an initial value of type $\typeTwo$, and constructs a function of type $\listt{}{\typeOne} \multimap \typeTwo$ implementing the fold of the step function over the input list. Although this workaround is not enough to recover the full power of general recursion, it appears that it is enough to describe many quantum algorithms. Figure \ref{fig:pq-fold} shows an example of the use of $\foldoperator$ to reverse a list. Note that $\abs{\tuple{\varOne}{\varTwo}}{\tensor{\typeOne}{\typeTwo}}{\termOne}$ is syntactic sugar for $\abs{\varThree}{\tensor{\typeOne}{\typeTwo}}{\dest{\varOne}{\varTwo}{\varThree}{\termOne}}$.

\begin{figure}[h]
	\centering
	\fbox{\parbox{.98\textwidth}{\centering
	\begin{align*}
		\mathit{rev} \triangleq 
		\fold{}{\lift{(\abs{\tuple{\mathit{revList}}{q}}{\tensor{\listt{}{\qubitt}}{\qubitt}}{\return{(\cons{q}{\mathit{revList}})}})}}{\nil}
\end{align*}}}
	\caption{Function \textit{rev} reverses a list of qubits.}
	\label{fig:pq-fold}
\end{figure}

To conclude this section, we just remark how all of the \Quipper\ programs shown in Section \ref{sec:overview} can be encoded in \PQ. However, \PQ's system of simple types in unable to tell us anything about the resource consumption of these programs. Of course, one could run \texttt{hadamardN} on a concrete input and examine the size of the circuit produced at run-time, but this amounts to \textit{testing}, not \textit{verifying} the program, and lacks the qualities of staticity and parametricity that we seek.

%% file: incepting-linear-dependency-and-effect-typing.tex
\label{sec:incepting-linear-dependency-and-effect-typing}

We are now ready to expand on the informal definition of the \PQ\ language given in Section \ref{sec:the-proto-quipper-language}, to reach a formal definition of \PQR: a linearly and dependently typed language whose type system supports the derivation of upper bounds on the width of the circuits produced by programs.

\subsection{Types and Syntax of \PQR}

\begin{figure}[!ht]
	\centering
	\fbox{\parbox{.98\textwidth}{\centering
	\begin{tabular}{l l r l}
		Types	& $\typeSet$	& $\typeOne,\typeTwo$ & $
		::= \unitt
		\mid \wtypeOne
		\mid \bang{\typeOne}
		\mid \tensor{\typeOne}{\typeTwo}
		\mid \arrowt{\typeOne}{\typeTwo}{\indexOne}{\indexTwo}
		\mid \listt{\indexOne}{\typeOne}
		\mid \circt{\indexOne}{\mtypeOne}{\mtypeTwo}$\\
		Param. types		& $\ptypeSet$	& $\ptypeOne,\ptypeTwo$ & $
		::= \unitt
		\mid \bang{\typeOne}
		\mid \tensor{\ptypeOne}{\ptypeTwo}
		\mid \listt{\indexOne}{\ptypeOne}
		\mid \circt{\indexOne}{\mtypeOne}{\mtypeTwo}$\\
		Bundle types	& $\mtypeset$	& $\mtypeOne,\mtypeTwo$ & $
		::= \unitt
		\mid \wtypeOne
		\mid \tensor{\mtypeOne}{\mtypeTwo}
		\mid \listt{\indexOne}{\mtypeOne},$\\
		&&&\\
		Terms	& $\termSet$	& $\termOne,\termTwo$ & $ ::= \app{\valOne}{\valTwo}
		\mid \dest{\varOne}{\varTwo}{\valOne}{\termOne}
		\mid \force{\valOne}
		\mid \boxt{\mtypeOne}{\valOne}$\\&&&$
		\mid \apply{\valOne}{\valTwo}
		\mid \return{\valOne}
		\mid \letin{\varOne}{\termOne}{\termTwo}$\\
		Values	& $\valSet$	& $\valOne,\valTwo$ & $
		::= \unitv
		\mid \varOne
		\mid \labOne
		\mid \abs{\varOne}{\typeOne}{\termOne}
		\mid \lift{\termOne}
		\mid \boxedCirc{\struct\labOne}{\circuitOne}{\struct{\labTwo}}
		\mid \tuple{\valOne}{\valTwo}$\\&&&$
		\mid \nil
		\mid \cons{\valOne}{\valTwo}
		\mid \fold{\ivarOne}{\valOne}{\valTwo}$\\
		Wire bundles	& $\mvalset$	& $\struct\labOne,\struct\labTwo$ & $
		::= \unitv
		\mid \labOne
		\mid \tuple{\struct\labOne}{\struct\labTwo}
		\mid \nil
		\mid \cons{\struct\labOne}{\struct\labTwo}$\\	
		Indices & $\indexSet$ & $\indexOne,\indexTwo$ & $
		::= \ivarOne
		\mid \natOne
		\mid \iplus{\indexOne}{\indexTwo}
		\mid \iminus{\indexOne}{\indexTwo}
		\mid \imult{\indexOne}{\indexTwo}
		\mid \imax{\indexOne}{\indexTwo}	
		\mid \imaximum{\ivarOne}{\indexOne}{\indexTwo}$
	\end{tabular}
	}}
	\caption{\PQR\ syntax and types.}
	\label{fig:pqr-syntax}
\end{figure}

The types and syntax of \PQR\ are given in Figure \ref{fig:pqr-syntax}.
As we mentioned, one of the key ingredients of our type system are the index terms which we annotate standard \PQ\ types with. These indices provide quantitative information about the elements of the resulting types, in a manner reminiscent of refinement types \cite{refinement-ml,liquid-types}. In our case, we are primarily concerned with circuit width, which means that the natural starting point of our extension of \PQ\ is precisely the circuit type $\circt{}{\mtypeOne}{\mtypeTwo}$: $\circt{\indexOne}{\mtypeOne}{\mtypeTwo}$ has elements the boxed circuits of input type $\mtypeOne$, output type $\mtypeTwo$, \textit{and width bounded by} $\indexOne$. Term $\indexOne$ is precisely what we call an index, that is, an arithmetical expression denoting a natural number. Looking at the grammar for indices, their interpretation is fairly straightforward: $\natOne$ is a natural number, $\ivarOne$ is an index variable, $\iplus{\indexOne}{\indexTwo},\imult{\indexOne}{\indexTwo}$ and $\imax{\indexOne}{\indexTwo}$ have their intuitive meaning, $\iminus{\indexOne}{\indexTwo}$ denotes \textit{natural} subtraction, and $\imaximum{\ivarOne}{\indexOne}{\indexTwo}$ is the maximum for $\ivarOne$ going from $0$ (included) to $\indexOne$ (excluded) of $\indexTwo$, where $\ivarOne$ can occur free in $\indexTwo$.

\begin{figure}[!ht]
	\fbox{\parbox{.98\textwidth}{\centering
	\begin{align*}
		\interpret{\wfjudgment{\icontextOne}{\indexOne}} :&\ \mathbb{N}^{|\icontextOne|} \to \mathbb{N} \\
		\interpret{\wfjudgment{\icontextOne}{n}} \left(n_1,\dots,n_{|\icontextOne|}\right) =&\ n \\
		\interpret{\wfjudgment{\icontextOne,\ivarOne}{\ivarOne}} \left(n_1,\dots,n_{|\icontextOne|},n_i\right) =&\ n_\ivarOne \\
		\interpret{\wfjudgment{\icontextOne}{\iplus{\indexOne}{\indexTwo}}} \left(n_1,\dots,n_{|\icontextOne|}\right) =&\ \interpret{\wfjudgment{\icontextOne}{\indexOne}}\left(n_1,\dots,n_{|\icontextOne|}\right) + \interpret{\wfjudgment{\icontextOne}{\indexTwo}}\left(n_1,\dots,n_{|\icontextOne|}\right) \\
		\interpret{\wfjudgment{\icontextOne}{\imult{\indexOne}{\indexTwo}}}  \left(n_1,\dots,n_{|\icontextOne|}\right) =&\ \interpret{\wfjudgment{\icontextOne}{\indexOne}}\left(n_1,\dots,n_{|\icontextOne|}\right)\interpret{\wfjudgment{\icontextOne}{\indexTwo}}\left(n_1,\dots,n_{|\icontextOne|}\right) \\
		\interpret{\wfjudgment{\icontextOne}{\iminus{\indexOne}{\indexTwo}}}  \left(n_1,\dots,n_{|\icontextOne|}\right) =&\ \max\left(0,\interpret{\wfjudgment{\icontextOne}{\indexOne}}\left(n_1,\dots,n_{|\icontextOne|}\right) - \interpret{\wfjudgment{\icontextOne}{\indexTwo}}\left(n_1,\dots,n_{|\icontextOne|}\right)\right) \\
		\interpret{\wfjudgment{\icontextOne}{\imax{\indexOne}{\indexTwo}}}  \left(n_1,\dots,n_{|\icontextOne|}\right) =&\ \max\left(\interpret{\wfjudgment{\icontextOne}{\indexOne}}\left(n_1,\dots,n_{|\icontextOne|}\right), \interpret{\wfjudgment{\icontextOne}{\indexTwo}}\left(n_1,\dots,n_{|\icontextOne|}\right)\right) \\
		\interpret{\wfjudgment{\icontextOne}{\imaximum{\ivarOne}{\indexOne}{\indexTwo}}} \left(n_1,\dots,n_{|\icontextOne|}\right) =&\ \max_{n\in\{0,\dots,\interpret{\wfjudgment{\icontextOne}{\indexOne}}\left(n_1,\dots,n_{|\icontextOne|}\right)-1\}} \interpret{\wfjudgment{\icontextOne,\ivarOne}{\indexTwo}}\left(n_1,\dots,n_{|\icontextOne|},n\right)
	\end{align*}
	}}
	\caption{Interpretation of well-formed indices.}
	\label{fig:index-semantics}
\end{figure}

Let $\icontextOne$ be a set of index variable names, which we call an \textit{index context}. An index $\indexOne$ is \emph{well-formed under context $\icontextOne$}, and we write $\wfjudgment{\icontextOne}{\indexOne}$, when all its free index variables are in $\icontextOne$. Figure \ref{fig:index-semantics} provides a more formal interpretation of well-formed indices.

While the index in a circuit type denotes an upper bound, the index in a type of the form $\listt{\indexOne}{\typeOne}$ denotes the \textit{exact} length of the lists of that type. While this refinement is quite restrictive in a generic scenario, it allows us to include lists of labels among wire bundles, since they are effectively isomorphic to finite tensors of labels and therefore represent wire bundles of known size.
Lastly, as we anticipated in Section \ref{sec:overview}, an arrow type $\arrowt{\typeOne}{\typeTwo}{\indexOne}{\indexTwo}$ is annotated with \textit{two} indices: $\indexOne$ is an upper bound to the width of the circuit built by the function once it is applied to an argument of type $\typeOne$, while $\indexTwo$ describes the exact number of wire resources captured in the function's closure. The utility of this last annotation will be clearer in Section \ref{sec:pqr-type-system}.

The languages for terms and values are almost the same as in \PQ, with the minor difference that the $\foldoperator$ operator now binds the index variable name $\ivarOne$ within the scope of its first argument. This variable appears locally in the type of the step function, in such a way as to allow each iteration of the fold to contribute to the overall circuit width in a \textit{different} way.

\subsection{A Formal Language for Circuits}
\label{sec:a-formal-language-for-circuits}

The type system of \PQR\ is designed to reason about the width of circuits. Therefore, before we formally introduce the type system in Section \ref{sec:pqr-type-system}, we ought to introduce circuits themselves in a formal way. So far, we have only spoken of circuits at a very high and intuitive level, and we have represented them only graphically. Looking at the circuits in Section \ref{sec:overview}, what do they have in common? At the fundamental level, they are made up of elementary operations applied to specific wires. Of course, the order of these operations matters, as does the order of wires that they are applied to (e.g. a controlled not operation does not have the same semantics if we switch the target and control qubits).

In the existing literature on \PQ, circuits are usually interpreted as morphisms in a symmetric monoidal category \cite{proto-quipper-m}, but this approach makes it particularly hard to reason about their intensional properties, such as width. For this reason, we opt for a \textit{concrete} model of wires and circuits, rather than an abstract one.

Luckily, we already have a datatype modeling ordered structures of wires, that is, the wire bundles that we introduced in the previous sections. We use them as the foundation upon which we build circuits.

\begin{figure}[!ht]
	\centering
	\fbox{\parbox{.98\textwidth}{\centering
		\begin{tabular}{l l r l}
			Wire bundles	& $\mvalset$	& $\struct\labOne,\struct\labTwo$ & $
			::= \unitv
			\mid \labOne
			\mid \tuple{\struct\labOne}{\struct\labTwo}
			\mid \nil
			\mid \cons{\struct\labOne}{\struct\labTwo},$\\
			Bundle types		& $\mtypeset$	& $\mtypeOne,\mtypeTwo$ & $
			::= \unitt \mid
			\wtypeOne
			\mid \tensor{\mtypeOne}{\mtypeTwo}
			\mid \listt{\indexOne}{\mtypeOne},$\\
			&&&\\
			Circuits	& $\circset$	& $\circuitOne,\circuitTwo$ & $ ::= \cidentity{\lcOne}
			\mid \circuitOne;\gateapp{\gateOne}{\struct\labOne}{\struct\labTwo}.$
			\end{tabular}
		}}
	\caption{\crl\ syntax and types.}
	\label{fig:crl-syntax}
\end{figure}

That being said, Figure \ref{fig:crl-syntax} introduces the \textsf{Circuit Representation Language} (\crl) which we use as the target for circuit building in \PQR. Wire bundles are exactly as in Figure \ref{fig:pqr-syntax} and represent arbitrary structures of wires, while circuits themselves are defined very simply as a sequence of elementary operations applied to said structures. We call $\lcOne$ a \textit{label context} and define it as a partial mapping from label names to wire types. We use label contexts as a mean to keep track of the set of labels available at any point during a computation, alongside their respective types. Circuit $\cidentity{\lcOne}$ represents the identity circuit taking as input the labels in $\lcOne$ and returning them unchanged, while $\circuitOne;\gateapp{\gateOne}{\struct\labOne}{\struct\labTwo}$ represents the application of the elementary operation $\gateOne$ to the wires identified by $\struct\labOne$ among the outputs of $\circuitOne$. Operation $\gateOne$ outputs the wire bundle $\struct\labTwo$, whose labels become part of the outputs of the circuit as a whole. Note that an ``elementary operation'' is usually the application of a gate, but it could also be a measurement, or the initialization or discarding of a wire. Although semantically very different, from the perspective of circuit building these operations are just elementary building blocks in the construction of a more complex structure, and it makes no sense to distinguish between them syntactically.
Circuits are amenable to a form of concatenation, defined as follows.

\begin{definition}[Circuit Concatenation]
	We define the \emph{concatenation of \crl\ circuits $\circuitOne$ and $\circuitTwo$}, written $\concat{\circuitOne}{\circuitTwo}$, as follows
	\begin{align}
		\concat{\circuitOne}{\cidentity{\lcOne}} &= \circuitOne \\
		\concat{\circuitOne}{(\circuitTwo;\gateapp{\gateOne}{\struct{\labOne}}{\struct{\labTwo}})} &= (\concat{\circuitOne}{\circuitTwo});\gateapp{\gateOne}{\struct{\labOne}}{\struct{\labTwo}}
	\end{align}
\end{definition}

\subsubsection{Circuit Typing}
Naturally, not all circuits built from the \crl\ grammar make sense. For example $\cidentity{(\labOne:\qubitt)};\gateapp{H}{\labTwo}{\labTwo}$ and $\cidentity{(\labOne:\qubitt)};\gateapp{\mathit{CNOT}}{\tuple{\labOne}{\labOne}}{\tuple{\labTwo}{\labThree}}$ are both syntactically correct, but the first applies a gate to a non-existing wire, while the second violates the no-cloning theorem by duplicating $\labOne$. To rule out such ill-formed circuits, we employ a rudimentary type system for circuits which allows us to derive judgments of the form $\circjudgment{\circuitOne}{\lcOne}{\lcTwo}$, which informally read ``circuit $\circuitOne$ is well-typed with input label context $\lcOne$ and output label context $\lcTwo$''.

\begin{figure}[!ht]
	\centering
	\fbox{\begin{mathpar}
			\inference[\textit{unit}]{}
			{\mjudgment{\emptycontext}{\unitv}{\unitt}}
			\and
			\inference[\textit{lab}]{}
			{\mjudgment{\labOne:\wtypeOne}{\labOne}{\wtypeOne}}
			\and
			\inference[\textit{nil}]
			{\eqjudgment{}{\indexOne}{0}}
			{\mjudgment{\emptycontext}{\nil}{\listt{\indexOne}{\mtypeOne}}}
			\and
			\inference[\textit{pair}]
			{\mjudgment{\lcOne_1}{\struct\labOne}{\mtypeOne}
				&
				\mjudgment{\lcOne_2}{\struct\labTwo}{\mtypeTwo}}
			{\mjudgment{\lcOne_1,\lcOne_2}{\tuple{\struct\labOne}{\struct\labTwo}}{\tensor{\mtypeOne}{\mtypeTwo}}}
			\and
			\inference[\textit{cons}]
			{\mjudgment{\lcOne_1}{\struct\labOne}{\mtypeOne}
				&
				\mjudgment{\lcOne_2}{\struct{\labTwo}}{\listt{\indexTwo}{\mtypeOne}}
				&
				\eqjudgment{}{\indexOne}{\iplus{\indexTwo}{1}}}
			{\mjudgment{\lcOne_1,\lcOne_2}{\cons{\struct\labOne}{\struct\labTwo}}{\listt{\indexOne}{\mtypeOne}}}
			\\\\
			\inference[\textit{id}]{}
			{\circjudgment{\cidentity{\lcOne}}{\lcOne}{\lcOne}}
			\and
			\inference[\textit{seq}]
			{\circjudgment{\circuitOne}{\lcOne}{\lcTwo,\lcThree}
				\quad
				\mjudgment{\lcThree}{\struct\labOne}{\mtypeOne}
				\quad
				\mjudgment{\lcFour}{\struct\labTwo}{\mtypeTwo}
				\quad
				\gateOne\in\gatesetoftype{\mtypeOne}{\mtypeTwo}}
			{\circjudgment{\circuitOne;\gateapp{\gateOne}{\struct\labOne}{\struct\labTwo}}{\lcOne}{\lcTwo,\lcFour}}
	\end{mathpar}}
	\caption{\crl\ type system.}
	\label{fig:crl-type-system}
\end{figure}

The typing rules for \crl\ are given in Figure \ref{fig:crl-type-system}. We call $\mjudgment{\lcOne}{\struct\labOne}{\mtypeOne}$ a \textit{wire judgment}, and we use it to give a structured type $\mtypeOne$ to an otherwise unordered label context $\lcOne$, by means of a wire bundle $\struct{\labOne}$. Most rules are straightforward, except those for lists, which rely on a judgment of the form $\eqjudgment{}{\indexOne}{\indexTwo}$. This is to be intended as a semantic judgment asserting that $\indexOne$ and $\indexTwo$ are closed and equal when interpreted as natural numbers. Within the typing rules for lists, this judgment reflects the idea that there are many ways to syntactically represent the length of a list. For example, $\nil$ can be given type $\listt{0}{\mtypeOne}$, but also $\listt{\iminus{1}{1}}{\mtypeOne}$ or $\listt{\imult{0}{5}}{\mtypeOne}$. This kind of flexibility might seem unwarranted for such a simple language, but it is useful to effectively interface \crl\ and the more complex \PQR. Speaking of the actual circuit judgments, the \textit{seq} rule tells us that the the application of an elementary operation $\gateOne$ is well-typed whenever $\gateOne$ only acts on labels occurring in the outputs of $\circuitOne$ (those in $\struct\labOne$, or equivalently in $\lcThree$), produces in output labels that do not clash with the remaining outputs of $\circuitOne$ (since $\lcTwo,\lcFour$ denotes the union of two label contexts with disjoint domains) and has the right signature. This last requirement is expressed as $\gateOne\in\gateset(\mtypeOne,\mtypeTwo)$, where $\gateset(\mtypeOne,\mtypeTwo)$ is the subset of elementary operations that can be applied to an input of type $\mtypeOne$ to obtain an output of type $\mtypeTwo$. For example, the Hadamard gate, which acts on a single qubit, is in $\gateset(\qubitt,\qubitt)$, the controlled not gate is in $\gateset(\tensor{\qubitt}{\qubitt},\tensor{\qubitt}{\qubitt})$ and the single-qubit measurement is in $\gateset(\qubitt,\bitt)$.

\subsubsection{Circuit Width}
Among the many properties of circuits, we are interested in width, so we conclude this section by giving a formal status to this quantity. As we saw in Section \ref{sec:overview}, when we initialize a new wire, we can reuse previously discarded wires in such a way that the width of a circuit is not always equal to the number of wires that are initialized. We formalize this intuition in the following definition.

\begin{definition}[Circuit Width]
	We define the \emph{width of a \crl\ circuit $\circuitOne$}, written $\width{\circuitOne}$, as follows
	\begin{align}
		\width{\cidentity{\lcOne}} &= |\lcOne|\\
		\width{\circuitOne;\gateapp{\gateOne}{\struct\labOne}{\struct\labTwo}} &= \width{\circuitOne} + \max(0, \inits{\gateOne} - \discarded{\circuitOne})
	\end{align}
	where $\discarded{\circuitOne} = \width{\circuitOne} - \operatorname{outputs}(\circuitOne)$ and
	\begin{align}
		\operatorname{outputs}(\cidentity{\lcOne}) &= |\lcOne|\\
		\operatorname{outputs}(\circuitOne;\gateapp{\gateOne}{\struct{\labOne}}{\struct{\labTwo}}) &= \operatorname{outputs}(\circuitOne) + \inits{\gateOne}
	\end{align}
\end{definition}
In the definition above, $|\lcOne|$ is the number of labels in $\lcOne$ and $\inits{\gateOne}$ represents the net number of new wires initialized by $\gateOne$. If $\gateOne$ outputs less wires than it consumes, then $\inits{\gateOne}$ is negative. The idea is that whenever we require a new wire in our computation, first we try to reuse a previously discarded wire, in which case the initialization does not add to the total width of the circuit ($\inits{\gateOne} \leq \discarded{\circuitOne}$), and \textit{only if we cannot do so} we actually create a new wire, increasing the overall width of the circuit ($\inits{\gateOne} > \discarded{\circuitOne}$).

Now that we have a formal definition of circuit types and width, we can state a fundamental property of the concatenation of well-typed circuits, which is illustrated in Figure \ref{fig:composition} and proven in Theorem \ref{thm:crl}. We use this result pervasively in proving the correctness of \PQR\ in section \ref{sec:type-safety-and-correctness}.

\begin{theorem}[CRL]
	Given $\circjudgment{\circuitOne}{\lcOne}{\lcTwo,\lcThree}$ and $\circjudgment{\circuitTwo}{\lcThree}{\lcFour}$ such that the labels shared by $\circuitOne$ and $\circuitTwo$ are all and only those in $\lcThree$, we have
	\begin{enumerate}
		\item $\circjudgment{\concat{\circuitOne}{\circuitTwo}}{\lcOne}{\lcTwo,\lcFour},$
		\item $\width{\concat{\circuitOne}{\circuitTwo}} \leq \maxf{\width{\circuitOne}}{\width{\circuitTwo} + |\lcTwo|}.$
	\end{enumerate}
	\label{thm:crl}
\end{theorem}
\begin{proof}
	By induction of the derivation of $\circjudgment{\circuitTwo}{\lcThree}{\lcFour}$.
\end{proof}

\begin{figure}[!ht]
	\centering
	\fbox{\parbox{.98\textwidth}{\centering
	\begin{quantikz}[row sep=1em]
		\lstick{$\lcOne$} & \qwbundle{|\lcOne|} & \gate[2]{\circuitOne} &  \ \push{\lcTwo} \ & \qwbundle{|\lcTwo|} & & & \rstick{$\lcTwo$}\\
		\setwiretype{n} & & & \setwiretype{q} \ \push{\lcThree} \  & \qwbundle{|\lcThree|} & \gate{\circuitTwo} & \qwbundle{|\lcFour|} &\rstick{$\lcFour$}
\end{quantikz}
}}
\caption{The concatenation of well-typed circuits $\circuitOne$ and $\circuitTwo$.}
\label{fig:composition}
\end{figure}

\subsection{Typing Programs}
\label{sec:pqr-type-system}

Going back to \PQR, we have already seen how the standard \PQ\ types are refined with quantitative information. However, decorating types is not enough for the purposes of width estimation. Recall that, in general, a \PQ\ program produces a circuit as a \textit{side effect} of its evaluation. If we want to reason about the width of said circuit, it is not enough to rely on a regular linear type system, although dependent. Rather, we have to introduce the second ingredient of our analysis and turn to a \textit{type-and-effect system} \cite{types-and-effects}, revolving around a type judgment of the form
\begin{equation}
	\cjudgment{\icontextOne}{\contextOne}{\lcOne}{\termOne}{\typeOne}{\indexOne},
\end{equation}
which intuitively reads ``for all values of the index variables in $\icontextOne$, under typing context $\contextOne$ and label context $\lcOne$, term $\termOne$ has type $\typeOne$ and produces a circuit of width at most $\indexOne$''. Therefore, the index variables in $\icontextOne$ are universally quantified in the rest of the judgment. Context $\contextOne$ is a typing context for parameter and linear variables alike. When a typing context contains exclusively parameter variables, we write it as $\pcontextOne$. In this judgment, index $\indexOne$ plays the role of an \textit{effect annotation}, describing a relevant aspect of the side effect produced by the evaluation of $\termOne$ (i.e. the width of the produced circuit). The attentive reader might wonder why this annotation consists only of one index, whereas when we discussed arrow types in previous sections we needed two. The reason is that the second index, which we use to keep track of the number of wires captured by a function, is redundant in a typing judgment where the same quantity can be inferred directly from the environments $\contextOne$ and $\lcOne$. A similar typing judgment is introduced for values, which are effect-less:
\begin{equation}
	\vjudgment{\icontextOne}{\contextOne}{\lcOne}{\valOne}{\typeOne}.
\end{equation}

\begin{figure}[p]
	\centering
	\fbox{\begin{mathpar}
			\inference[\textit{unit}]
			{\wfjudgment{\icontextOne}{\pcontextOne}}
			{\vjudgment{\icontextOne}{\pcontextOne}{\emptylc}{\unitv}{\unitt}}
			\and
			\inference[\textit{lab}]
			{\wfjudgment{\icontextOne}{\pcontextOne}}
			{\vjudgment{\icontextOne}{\pcontextOne}{\labOne:\wtypeOne}{\labOne}{\wtypeOne}}
			\and
			\inference[\textit{var}]
			{\wfjudgment{\icontextOne}{\pcontextOne,\varOne:\typeOne}}
			{\vjudgment{\icontextOne}{\pcontextOne,\varOne:\typeOne}{\emptylc}{\varOne}{\typeOne}}
			\and
			\inference[\textit{abs}]
			{\cjudgment{\icontextOne}{\contextOne,\varOne:\typeOne}{\lcOne}{\termOne}{\typeTwo}{\indexOne}}
			{\vjudgment{\icontextOne}{\contextOne}{\lcOne}{\abs{\varOne}{\typeOne}{\termOne}}{\arrowt{\typeOne}{\typeTwo}{\indexOne}{\rcount{\contextOne;\lcOne}}}}
			\and
			\inference[\textit{app}]
			{\vjudgment{\icontextOne}{\pcontextOne,\contextOne_1}{\lcOne_1}{\valOne}{\arrowt{\typeOne}{\typeTwo}{\indexOne}{\indexTwo}}
				&
				\vjudgment{\icontextOne}{\pcontextOne,\contextOne_2}{\lcOne_2}{\valTwo}{\typeOne}}
			{\cjudgment{\icontextOne}{\pcontextOne,\contextOne_1,\contextOne_2}{\lcOne_1,\lcOne_2}{\app{\valOne}{\valTwo}}{\typeTwo}{\indexOne}}
			\and
			\inference[\textit{lift}]
			{\cjudgment{\icontextOne}{\pcontextOne}{\emptycontext}{\termOne}{\typeOne}{0}}
			{\vjudgment{\icontextOne}{\pcontextOne}{\emptycontext}{\lift{\termOne}}{\bang{\typeOne}}}
			\and
			\inference[\textit{force}]
			{\vjudgment{\icontextOne}{\pcontextOne}{\emptycontext}{\valOne}{\bang{\typeOne}}}
			{\cjudgment{\icontextOne}{\pcontextOne}{\emptycontext}{\force{\valOne}}{\typeOne}{0}}
			\and
			\inference[\textit{circ}]
			{\circjudgment{\circuitOne}{\lcOne}{\lcTwo}
				&
				\mjudgment{\lcOne}{\struct\labOne}{\mtypeOne}
				&
				\mjudgment{\lcTwo}{\struct\labTwo}{\mtypeTwo}
				&
				\leqjudgment{\icontextOne}{\width{\circuitOne}}{\indexOne}
				&
				\wfjudgment{\icontextOne}{\pcontextOne}}
			{\vjudgment{\icontextOne}{\pcontextOne}{\emptycontext}{\boxedCirc{\struct\labOne}{\circuitOne}{\struct\labTwo}}{\circt{\indexOne}{\mtypeOne}{\mtypeTwo}}}
			\and
			\inference[\textit{apply}]
			{\vjudgment{\icontextOne}{\pcontextOne,\contextOne_1}{\lcOne_1}{\valOne}{\circt{\indexOne}{\mtypeOne}{\mtypeTwo}}
				&
				\vjudgment{\icontextOne}{\pcontextOne,\contextOne_2}{\lcOne_2}{\valTwo}{\mtypeOne}}
			{\cjudgment{\icontextOne}{\pcontextOne,\contextOne_1,\contextOne_2}{\lcOne_1,\lcOne_2}{\apply{\valOne}{\valTwo}}{\mtypeTwo}{\indexOne}}
			\and
			\inference[\textit{box}]
			{\vjudgment{\icontextOne}{\pcontextOne}{\emptycontext}{\valOne}{\bang({\arrowt{\mtypeOne}{\mtypeTwo}{\indexOne}{\indexTwo}})}}
			{\cjudgment{\icontextOne}{\pcontextOne}{\emptycontext}{\boxt{\mtypeOne}{\valOne}}{\circt{\indexOne}{\mtypeOne}{\mtypeTwo}}{0}}
			\and
			\inference[\textit{nil}]
			{\wfjudgment{\icontextOne}{\pcontextOne}
				&
				\wfjudgment{\icontextOne}{\typeOne}}
			{\vjudgment{\icontextOne}{\pcontextOne}{\emptycontext}{\nil}{\listt{0}{\typeOne}}}
			\and
			\inference[\textit{cons}]
			{\vjudgment{\icontextOne}{\pcontextOne,\contextOne_1}{\lcOne_1}{\valOne}{\typeOne}
				&
				\vjudgment{\icontextOne}{\pcontextOne,\contextOne_2}{\lcOne_2}{\valTwo}{\listt{\indexOne}{\typeOne}}}
			{\vjudgment{\icontextOne}{\pcontextOne,\contextOne_1,\contextOne_2}{\lcOne_1\lcOne_2}{\cons{\valOne}{\valTwo}}{\listt{\iplus{\indexOne}{1}}{\typeOne}}}
			\and
			\inference[\textit{fold}]
			{\vjudgment{\icontextOne}{\pcontextOne,\contextOne}{\lcOne}{\valTwo}{\typeTwo\isub{0}{\ivarOne}}
				&
				\vjudgment{\icontextOne,\ivarOne}{\pcontextOne}{\emptycontext}{\valOne}{\bang{(\arrowt{(\tensor{\typeTwo}{\typeOne})}{\typeTwo\isub{\iplus{\ivarOne}{1}}{\ivarOne}}{\indexTwo}{\indexTwo'})}}
				\\
				\wfjudgment{\icontextOne}{\indexOne}
				&
				\wfjudgment{\icontextOne}{\typeOne}
				&
				\indexThree = \imax{\rcount{\contextOne;\lcOne}}{\imaximum{\ivarOne}{\indexOne}{\iplus{\indexTwo}{\imult{(\iminus{\iminus{\indexOne}{1}}{\ivarOne})}{\rcount{\typeOne}}}}}}
			{\vjudgment{\icontextOne}{\pcontextOne,\contextOne}{\lcOne}{\fold{\ivarOne}{\valOne}{\valTwo}}{\arrowt{\listt{\indexOne}{\typeOne}}{\typeTwo\isub{\indexOne}{\ivarOne}}{\indexThree}{\rcount{\contextOne;\lcOne}}}}
			\and
			\inference[\textit{dest}]
			{\vjudgment{\icontextOne}{\pcontextOne,\contextOne_1}{\lcOne_1}{\valOne}{\tensor{\typeOne}{\typeTwo}}
				&
				\cjudgment{\icontextOne}{\pcontextOne,\contextOne_2,\varOne:\typeOne,\varTwo:\typeTwo}{\lcOne_2}{\termOne}{\typeThree}{\indexOne}}
			{\cjudgment{\icontextOne}{\pcontextOne,\contextOne_1,\contextOne_2}{\lcOne_1,\lcOne_2}{\dest{\varOne}{\varTwo}{\valOne}{\termOne}}{\typeThree}{\indexOne}}
			\and
			\inference[\textit{pair}]
			{\vjudgment{\icontextOne}{\pcontextOne,\contextOne_1}{\lcOne_1}{\valOne}{\typeOne}
				&
				\vjudgment{\icontextOne}{\pcontextOne,\contextOne_2}{\lcOne_2}{\valTwo}{\typeTwo}}
			{\vjudgment{\icontextOne}{\pcontextOne,\contextOne_1,\contextOne_2}{\lcOne_1,\lcOne_2}{\tuple{\valOne}{\valTwo}}{\tensor{\typeOne}{\typeTwo}}}
			\and
			\inference[\textit{return}]
			{\vjudgment{\icontextOne}{\contextOne}{\lcOne}{\valOne}{\typeOne}}
			{\cjudgment{\icontextOne}{\contextOne}{\lcOne}{\return{\valOne}}{\typeOne}{\rcount{\contextOne;\lcOne}}}
			\and
			\inference[\textit{let}]
			{\cjudgment{\icontextOne}{\pcontextOne,\contextOne_1}{\lcOne_1}{\termOne}{\typeOne}{\indexOne}
				&
				\cjudgment{\icontextOne}{\pcontextOne,\contextOne_2,\varOne:\typeOne}{\lcOne_2}{\termTwo}{\typeTwo}{\indexTwo}}
			{\cjudgment{\icontextOne}{\pcontextOne,\contextOne_1,\contextOne_2}{\lcOne_1,\lcOne_2}{\letin{\varOne}{\termOne}{\termTwo}}{\typeTwo}{\imax{\iplus{\indexOne}{\rcount{\contextOne_2;\lcOne_2}}}{\indexTwo}}}
			\and
			\inference[\textit{vsub}]
			{\vjudgment{\icontextOne}{\contextOne}{\lcOne}{\valOne}{\typeOne}
				&
				\subtypejudgment{\icontextOne}{\typeOne}{\typeTwo}}
			{\vjudgment{\icontextOne}{\contextOne}{\lcOne}{\valOne}{\typeTwo}}
			\and
			\inference[\textit{csub}]
			{\cjudgment{\icontextOne}{\contextOne}{\lcOne}{\termOne}{\typeOne}{\indexOne}
				&
				\subtypejudgment{\icontextOne}{\typeOne}{\typeTwo}
				&
				\leqjudgment{\icontextOne}{\indexOne}{\indexTwo}}
			{\cjudgment{\icontextOne}{\contextOne}{\lcOne}{\termOne}{\typeTwo}{\indexTwo}}
	\end{mathpar}}
	\caption{\PQR\ type system.}
	\label{fig:typing-rules}
\end{figure}

\begin{figure}[!ht]
	\centering
	\fbox{\begin{mathpar}
			\inference[ivar]
			{ }
			{\wfjudgment{\icontextOne,\ivarOne}{\ivarOne}}
			\and
			\inference[nat]
			{ }
			{\wfjudgment{\icontextOne}{\natOne}}
			\and
			\inference[plus]
			{\wfjudgment{\icontextOne}{\indexOne}
				&
				\wfjudgment{\icontextOne}{\indexTwo}}
			{\wfjudgment{\icontextOne}{\iplus{\indexOne}{\indexTwo}}}
			\and
			\inference[minus]
			{\wfjudgment{\icontextOne}{\indexOne}
				&
				\wfjudgment{\icontextOne}{\indexTwo}}
			{\wfjudgment{\icontextOne}{\iminus{\indexOne}{\indexTwo}}}
			\and
			\inference[mult]
			{\wfjudgment{\icontextOne}{\indexOne}
				&
				\wfjudgment{\icontextOne}{\indexTwo}}
			{\wfjudgment{\icontextOne}{\imult{\indexOne}{\indexTwo}}}
			\and
			\inference[max]
			{\wfjudgment{\icontextOne}{\indexOne}
				&
				\wfjudgment{\icontextOne}{\indexTwo}}
			{\wfjudgment{\icontextOne}{\imax{\indexOne}{\indexTwo}}}
			\and
			\inference[maximum]
			{\wfjudgment{\icontextOne}{\indexOne}
				&
				\wfjudgment{\icontextOne,\ivarOne}{\indexTwo}}
			{\wfjudgment{\icontextOne}{\imaximum{\ivarOne}{\indexOne}{\indexTwo}}}
			\\\\
			\inference[unit]
			{}
			{\wfjudgment{\icontextOne}{\unitt}}
			\and
			\inference[wire]
			{}
			{\wfjudgment{\icontextOne}{\wtypeOne}}
			\and
			\inference[bang]
			{\wfjudgment{\icontextOne}{\typeOne}}
			{\wfjudgment{\icontextOne}{\bang{\typeOne}}}
			\and
			\inference[tensor]
			{\wfjudgment{\icontextOne}{\typeOne}
				&
				\wfjudgment{\icontextOne}{\typeTwo}}
			{\wfjudgment{\icontextOne}{\tensor{\typeOne}{\typeTwo}}}
			\and
			\inference[arrow]
			{\wfjudgment{\icontextOne}{\typeOne}
				&
				\wfjudgment{\icontextOne}{\typeTwo}
				&
				\wfjudgment{\icontextOne}{\indexOne}
				&
				\wfjudgment{\icontextOne}{\indexTwo}}
			{\wfjudgment{\icontextOne}{\arrowt{\typeOne}{\typeTwo}{\indexOne}{\indexTwo}}}
			\and
			\inference[list]
			{\wfjudgment{\icontextOne}{\typeOne}
				&
				\wfjudgment{\icontextOne}{\indexOne}}
			{\wfjudgment{\icontextOne}{\listt{\indexOne}{\typeOne}}}
			\and
			\inference[circ]
			{\wfjudgment{\icontextOne}{\mtypeOne}
				&
				\wfjudgment{\icontextOne}{\mtypeTwo}
				&
				\wfjudgment{\icontextOne}{\indexOne}}
			{\wfjudgment{\icontextOne}{\circt{\indexOne}{\mtypeOne}{\mtypeTwo}}}
	\end{mathpar}}
\caption{\PQR\ well-formedness rules.}
\label{fig:well-formedness}
\end{figure}

The rules for deriving typing judgments are those in Figure \ref{fig:typing-rules}, where $\contextOne_1,\contextOne_2$ and $\lcOne_1,\lcOne_2$ denote the union of two contexts with disjoint domains. The well-formedness judgment $\wfjudgment{\icontextOne}{\indexOne}$ is extended to types as shown in Figure \ref{fig:well-formedness} and then lifted to typing contexts in the natural way. Among interesting typing rules, we can see how the \textit{circ} rule bridges between \crl\ and \PQR. A boxed circuit $\boxedCirc{\struct\labOne}{\circuitOne}{\struct\labTwo}$ is well typed with type $\circt{\indexOne}{\mtypeOne}{\mtypeTwo}$ when $\circuitOne$ is no wider than the quantity denoted by $\indexOne$, $\circjudgment{\circuitOne}{\lcOne}{\lcTwo}$ and $\struct\labOne,\struct\labTwo$ contain all and only the labels in $\lcOne$ and $\lcTwo$, respectively, acting as a language-level interface to $\circuitOne$.

The two main constructs that interact with circuits are $\applyoperator$ and $\boxoperator$. The \textit{apply} rule is the foremost place where effects enter the type derivation:

\begin{equation*}
\inference[\textit{apply}]
{\vjudgment{\icontextOne}{\pcontextOne,\contextOne_1}{\lcOne_1}{\valOne}{\circt{\indexOne}{\mtypeOne}{\mtypeTwo}}
	&
	\vjudgment{\icontextOne}{\pcontextOne,\contextOne_2}{\lcOne_2}{\valTwo}{\mtypeOne}}
{\cjudgment{\icontextOne}{\pcontextOne,\contextOne_1,\contextOne_2}{\lcOne_1,\lcOne_2}{\apply{\valOne}{\valTwo}}{\mtypeTwo}{\indexOne}}
\end{equation*}

Since $\valOne$ represents some boxed circuit of width at most $\indexOne$, its application to an appropriate wire bundle $\valTwo$ produces exactly a circuit of width at most $\indexOne$.
The \textit{box} rule, on the other hand, works more or less in the opposite direction:

\begin{equation*}
\inference[\textit{box}]
{\vjudgment{\icontextOne}{\pcontextOne}{\emptycontext}{\valOne}{\bang({\arrowt{\mtypeOne}{\mtypeTwo}{\indexOne}{\indexTwo}})}}
{\cjudgment{\icontextOne}{\pcontextOne}{\emptycontext}{\boxt{\mtypeOne}{\valOne}}{\circt{\indexOne}{\mtypeOne}{\mtypeTwo}}{0}}
\end{equation*}

If $\valOne$ is a circuit building function that, once applied to an input of type $\mtypeOne$, would build a circuit of output type $\mtypeTwo$ and width at most $\indexOne$, then boxing it means turning it into a boxed circuit with the same characteristics. Note that the \textit{box} rule requires that the typing context be devoid of linear variables. This reflects the idea that $\valOne$ is meant to be executed in complete isolation, to build a standalone, replicable circuit, and therefore it should not capture any linear resource (e.g. a label) from the surrounding environment.

\subsubsection{Wire Count}
Notice that many rules rely on an operator written $\rcount{\cdot}$, which we call the \textit{wire count} operator. Intuitively, this operator returns the number of wire resources (in our case, bits or qubits) represented by a type or context. To understand how this is important, consider the \textit{return} rule:

\begin{equation*}
\inference[\textit{return}]
{\vjudgment{\icontextOne}{\contextOne}{\lcOne}{\valOne}{\typeOne}}
{\cjudgment{\icontextOne}{\contextOne}{\lcOne}{\return{\valOne}}{\typeOne}{\rcount{\contextOne;\lcOne}}}
\end{equation*}

The $\returnoperator$ operator turns a value $\valOne$ into a trivial computation that evaluates immediately to $\valOne$, and therefore it would be tempting to give it an effect annotation of $0$. However, $\valOne$ is not necessarily a closed value. In fact, it might very well contain many bits and qubits, coming both from the typing context $\contextOne$ and the label context $\lcOne$. Although nothing happens to these bits and qubits, they still corresponds to wires in the underlying circuit, and these wires have a width which must be accounted for in the judgment for the otherwise trivial computation. The \textit{return} rule therefore produces an effect annotation of the form $\rcount{\contextOne;\lcOne}$, which corresponds exactly to this quantity. A formal description of the wire count operator on types is given in the following Definition \ref{def:resource-count}.

\begin{definition}[Wire Count]
	\label{def:resource-count}
	We define the \emph{wire count of a type $\typeOne$}, written $\rcount{\typeOne}$, as a function $\rcount{\cdot}:\typeSet\to\indexSet$ such that
	\begin{mathpar}
		\rcount{\unitt} = \rcount{\bang{\typeOne}} = \rcount{\circt{\indexOne}{\mtypeOne}{\mtypeTwo}} =0
		\and
		\rcount{\wtypeOne} = 1
		\and
		\rcount{\tensor{\typeOne}{\typeTwo}} = \iplus{\rcount{\typeOne}}{\rcount{\typeTwo}}
		\and
		\rcount{\arrowt{\typeOne}{\typeTwo}{\indexOne}{\indexTwo}} = \indexTwo
		\and
		\rcount{\listt{\indexOne}{\typeOne}} = \imult{\indexOne}{\rcount{\typeOne}}
	\end{mathpar}
\end{definition}

This definition is lifted to typing and label contexts in the natural way. Note that, for any label context $\lcOne$, we have $\rcount{\lcOne}=|\lcOne|$. Annotation $\rcount{\contextOne;\lcOne}$ is then shorthand for $\iplus{\rcount{\contextOne}}{|\lcOne|}$. This definition is fairly straightforward, except for the arrow case. By itself, an arrow type does not give us any information about the amount of qubits or bits captured in the corresponding closure. This is precisely where the second index $\indexTwo$, which keeps track exactly of this quantity, comes into play. This annotation is introduced by the \textit{abs} rule and allows our analysis to circumvent data hiding.

The \textit{let} rule is another rule in which wire counts are essential:
\begin{equation*}
\inference[\textit{let}]
{\cjudgment{\icontextOne}{\pcontextOne,\contextOne_1}{\lcOne_1}{\termOne}{\typeOne}{\indexOne}
	&
	\cjudgment{\icontextOne}{\pcontextOne,\contextOne_2,\varOne:\typeOne}{\lcOne_2}{\termTwo}{\typeTwo}{\indexTwo}}
{\cjudgment{\icontextOne}{\pcontextOne,\contextOne_1,\contextOne_2}{\lcOne_1,\lcOne_2}{\letin{\varOne}{\termOne}{\termTwo}}{\typeTwo}{\imax{\iplus{\indexOne}{\rcount{\contextOne_2;\lcOne_2}}}{\indexTwo}}}
\end{equation*}
The two terms $\termOne$ and $\termTwo$ build the circuits $\circuitOne_\termOne$ and $\circuitOne_\termTwo$, whose widths are bounded by $\indexOne$ and $\indexTwo$, respectively. Once again, it might be tempting to conclude that the overall circuit built by the let construct has width bounded by $\imax{\indexOne}{\indexTwo}$, but this fails to take into account the fact that while $\termOne$ is building $\circuitOne_\termOne$ starting from the wires contained in $\contextOne_1$ and $\lcOne_1$, we must keep aside the wires contained in $\contextOne_2$ and $\lcOne_2$, which will be used by $\termTwo$ to build $\circuitOne_\termTwo$. These wires must flow alongside $\circuitOne_\termOne$ and their width, i.e. $\rcount{\contextOne_2;\lcOne_2}$, adds up to the total width of the left-hand side of the $\letoperator$ construct, leading to an overall width upper bound of $\imax{\indexOne+\rcount{\contextOne_2;\lcOne_2}}{\indexTwo}$. This situation is better illustrated in Figure \ref{fig:let-circuit}.

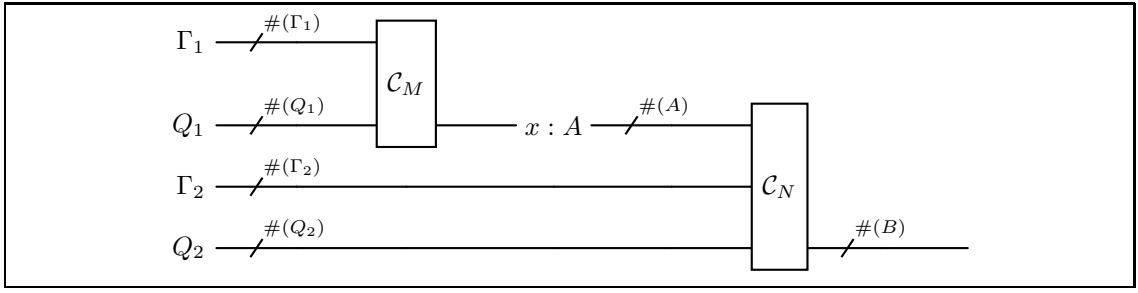
\begin{figure}[!ht]
	\centering
	\fbox{\parbox{.98\textwidth}{\centering
	\begin{quantikz}[row sep=1.5em, column sep=3em]
		\lstick{$\contextOne_1$}& \qwbundle{\rcount{\contextOne_1}} & \gate[2]{\circuitOne_\termOne} & \setwiretype{n}\\
		\lstick{$\lcOne_1$}& \qwbundle{\rcount{\lcOne_1}}& &\ \push{x:\typeOne} \ & \qwbundle{\rcount{\typeOne}} & \gate[3]{\circuitOne_\termTwo}\\
		\lstick{$\contextOne_2$}& \qwbundle{\rcount{\contextOne_2}}&  & & &\\
		\lstick{$\lcOne_2$}& \qwbundle{\rcount{\lcOne_2}}& & &	& & \qwbundle{\rcount{\typeTwo}} &
\end{quantikz}
}}
\caption{The shape of a circuit built by a $\letoperator$ construct.}
\label{fig:let-circuit}
\end{figure}

The last rule that makes substantial use of wire counts is \textit{fold}, arguably the most complex rule in the system:

\begin{equation*}
\inference[\textit{fold}]
{\vjudgment{\icontextOne}{\pcontextOne,\contextOne}{\lcOne}{\valTwo}{\typeTwo\isub{0}{\ivarOne}}
	&
	\vjudgment{\icontextOne,\ivarOne}{\pcontextOne}{\emptycontext}{\valOne}{\bang{(\arrowt{(\tensor{\typeTwo}{\typeOne})}{\typeTwo\isub{\iplus{\ivarOne}{1}}{\ivarOne}}{\indexTwo}{\indexTwo'})}}
	\\
	\wfjudgment{\icontextOne}{\indexOne}
	&
	\wfjudgment{\icontextOne}{\typeOne}
	&
	\indexThree = \imax{\rcount{\contextOne;\lcOne}}{\imaximum{\ivarOne}{\indexOne}{\iplus{\indexTwo}{\imult{(\iminus{\iminus{\indexOne}{1}}{\ivarOne})}{\rcount{\typeOne}}}}}}
{\vjudgment{\icontextOne}{\pcontextOne,\contextOne}{\lcOne}{\fold{\ivarOne}{\valOne}{\valTwo}}{\arrowt{\listt{\indexOne}{\typeOne}}{\typeTwo\isub{\indexOne}{\ivarOne}}{\indexThree}{\rcount{\contextOne;\lcOne}}}}
\end{equation*}

The main ingredient of the fold rule is the bound index variable $\ivarOne$, which occurs in the accumulator type $\typeTwo$ and is used to keep track of the number of steps performed by the fold. Let $(\cdot)\isub{\indexOne}{\ivarOne}$ denote the capture-avoiding substitution of the index term $\indexOne$ for the index variable $\ivarOne$ inside an index, type, context, value or term, not unlike $(\cdot)\sub{\valOne}{\varOne}$ denotes the capture-avoiding substitution of the value $\valOne$ for the variable $\varOne$. Intuitively, if the accumulator has initially type $\typeTwo\isub{0}{\ivarOne}$ and each application of the step function increases $\ivarOne$ by one, then when we fold over a list of length $\indexOne$ we get an output of type $\typeTwo\isub{\indexOne}{\ivarOne}$. Index $\indexThree$ is the upper bound to the width of the overall circuit built by the fold: if the input list is empty, then the width of the circuit is just the number of wires contained in the initial accumulator, that is, $\rcount{\contextOne;\lcOne}$. If the input list is non-empty, on the other hand, things get slightly more complicated. At each step $\ivarOne$, the step function builds a circuit $\circuitOne_\ivarOne$ of width bounded by $\indexTwo$, where $\indexTwo$ might depend on $\ivarOne$. This circuit takes as input all the wires in the accumulator, as well as the wires contained in the first element of the input list, which are $\rcount{\typeOne}$. The wires contained in remaining $\indexOne - 1 - \ivarOne$ elements have to flow alongside $\circuitOne_\ivarOne$, giving a width upper bound of $\indexTwo + (\indexOne - 1 - \ivarOne) \times \rcount{\typeOne}$ at each step $\ivarOne$. The overall width upper bound is then the maximum for $\ivarOne$ going from $0$ to $\indexOne-1$ of this quantity, i.e. precisely $\imaximum{\ivarOne}{\indexOne}{\indexTwo + (\indexOne - 1 - \ivarOne) \times \rcount{\typeOne}}$. Once again, a graphical representation of this scenario is given in Figure \ref{fig:fold-circuit}.

\begin{figure}[!ht]
	\centering
	\fbox{\parbox{.98\textwidth}{\centering
	\begin{quantikz}[column sep=2em, row sep=1em]
		\lstick{$\contextOne$} & \qwbundle{\rcount{\contextOne}} & \gate[3]{\circuitOne_0}\\
		\lstick{$\lcOne$} & \qwbundle{\rcount{\lcOne}} &  \\
		& \qwbundle{\rcount{\typeOne}} & &\qwbundle{\rcount{\typeTwo\isub{1}{\ivarOne}}} & & \gate[2]{\circuitOne_1}\\
		& \qwbundle{\rcount{\typeOne}} & & & & & \qwbundle{\rcount{\typeTwo\isub{2}{\ivarOne}}} & \ \ldots \ &\setwiretype{n}&&&&&&\\
		\wave&&&&&&&&&&&&&& \\
		\setwiretype{n}&&&&&&&\ldots \ &\setwiretype{q}\qwbundle{\rcount{\typeTwo\isub{\indexOne-1}{\ivarOne}}} &&& \gate[2]{\circuitOne_{\indexOne-1}}\\
		&\qwbundle{\rcount{\typeOne}} &&&&&&&&&&&\qwbundle{\rcount{\typeTwo\isub{\indexOne}{\ivarOne}}}&&\rstick{}
\end{quantikz}
}}
\caption{The shape of a circuit built by a fold applied to an input list of type $\listt{\indexOne}{\typeOne}$.}
\label{fig:fold-circuit}
\end{figure}
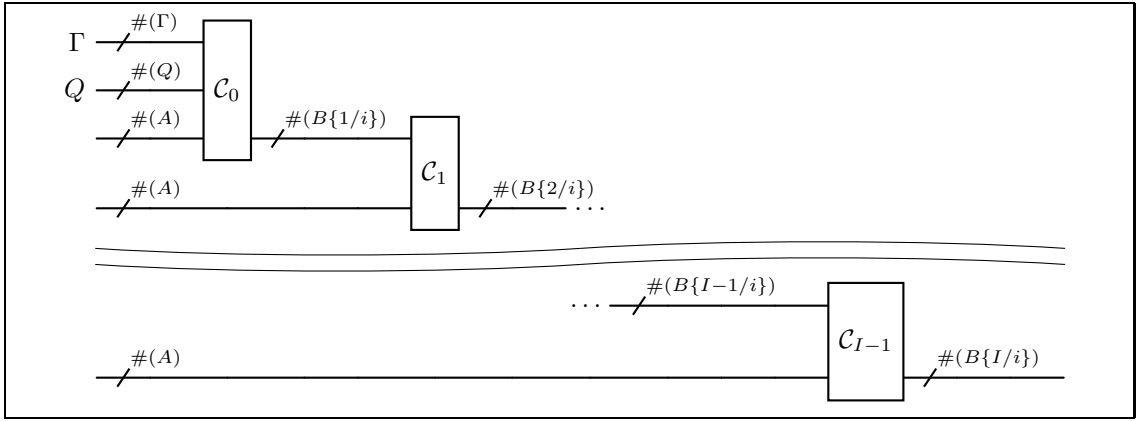

\subsubsection{Subtyping}

Notice that \PQR's type system includes two rules for subtyping, which are effectively the same rule for terms and values, respectively: \textit{csub} and \textit{vsub}.
We mentioned that our type system resembles a refinement type system, and all such systems induce a subtyping relation between types, where $\typeOne$ is a subtype of $\typeTwo$ whenever the former is ``at least as refined'' as the latter. In our case, a subtyping judgment such as $\subtypejudgment{\icontextOne}{\typeOne}{\typeTwo}$ means that for all natural values of the index variables in $\icontextOne$, $\typeOne$ is a subtype of $\typeTwo$.

\begin{figure}[!ht]
	\centering
	\fbox{\begin{mathpar}
			\inference[\textit{unit}]
			{}
			{\subtypejudgment{\icontextOne}{\unitt}{\unitt}}
			\and
			\inference[\textit{wire}]
			{}
			{\subtypejudgment{\icontextOne}{\wtypeOne}{\wtypeOne}}
			\and
			\inference[\textit{bang}]
			{\subtypejudgment{\icontextOne}{\typeOne}{\typeTwo}}
			{\subtypejudgment{\icontextOne}{\bang{\typeOne}}{\bang{\typeTwo}}}
			\and
			\inference[\textit{tensor}]
			{\subtypejudgment{\icontextOne}{\typeOne_1}{\typeOne_2}
				&
				\subtypejudgment{\icontextOne}{\typeTwo_1}{\typeTwo_2}}
			{\subtypejudgment{\icontextOne}{\tensor{\typeOne_1}{\typeTwo_1}}{\tensor{\typeOne_2}{\typeTwo_2}}}
			\and
			\inference[\textit{arrow}]
			{\subtypejudgment{\icontextOne}{\typeOne_2}{\typeOne_1}
				&
				\subtypejudgment{\icontextOne}{\typeTwo_1}{\typeTwo_2}
				&
				\leqjudgment{\icontextOne}{\indexOne_1}{\indexOne_2}
				&
				\eqjudgment{\icontextOne}{\indexTwo_1}{\indexTwo_2}}
			{\subtypejudgment{\icontextOne}{\arrowt{\typeOne_1}{\typeTwo_1}{\indexOne_1}{\indexTwo_1}}{\arrowt{\typeOne_2}{\typeTwo_2}{\indexOne_2}{\indexTwo_2}}}
			\and
			\inference[\textit{list}]
			{\subtypejudgment{\icontextOne}{\typeOne}{\typeTwo}
				&
				\eqjudgment{\icontextOne}{\indexOne}{\indexTwo}}
			{\subtypejudgment{\icontextOne}{\listt{\indexOne}{\typeOne}}{\listt{\indexTwo}{\typeTwo}}}
			\and
			\inference[\textit{circ}]
			{\sametypejudgment{\icontextOne}{\mtypeOne_1}{\mtypeOne_2}
				&
				\sametypejudgment{\icontextOne}{\mtypeTwo_1}{\mtypeTwo_2}
				&
				\leqjudgment{\icontextOne}{\indexOne}{\indexTwo}}
			{\subtypejudgment{\icontextOne}{\circt{\indexOne}{\mtypeOne_1}{\mtypeTwo_1}}{\circt{\indexTwo}{\mtypeOne_2}{\mtypeTwo_2}}}
	\end{mathpar}}
	\caption{\PQR\ subtyping rules.}
	\label{fig:subtyping-rules}
\end{figure}

We derive this kind of judgments by the rules in Figure \ref{fig:subtyping-rules}. Note that $\sametypejudgment{\icontextOne}{\typeOne}{\typeTwo}$ is shorthand for ``$\subtypejudgment{\icontextOne}{\typeOne}{\typeTwo}$ and $\subtypejudgment{\icontextOne}{\typeTwo}{\typeOne}$''. Subtyping relies in turn on a judgment of the form $\leqjudgment{\icontextOne}{\indexOne}{\indexTwo}$, which is a generalization of the semantic judgment that we used in the \crl\ type system in Section \ref{sec:a-formal-language-for-circuits}. Such a judgment asserts that for all values of the index variables in $\icontextOne$, $\indexOne$ is lesser or equal than $\indexTwo$. More formally, the meaning of $\leqjudgment{\icontextOne}{\indexOne}{\indexTwo}$ is that $\wfjudgment{\icontextOne}{\indexOne},\wfjudgment{\icontextOne}{\indexTwo}$ and for all $n_1,\dots,n_{|\icontextOne|}$: $\interpret{\wfjudgment{\icontextOne}{\indexOne}}(n_1,\dots,n_{|\icontextOne|}) \leq \interpret{\wfjudgment{\icontextOne}{\indexTwo}}(n_1,\dots,n_{|\icontextOne|})$.
Consequently, $\eqjudgment{}{\indexOne}{\indexTwo}$ is just shorthand for $\eqjudgment{\emptycontext}{\indexOne}{\indexTwo}$, which in turn is shorthand for ``$\leqjudgment{\icontextOne}{\indexOne}{\indexTwo}$ and $\leqjudgment{\icontextOne}{\indexTwo}{\indexOne}$''. We purposefully leave the decision procedure of this kind of judgments unspecified, with the prospect that, in a practical scenario, they could be delegated to an SMT solver \cite{smt-solvers}.

\subsection{Operational Semantics}
\label{sec:pqr-operational-semantics}

Operationally speaking, it does not make sense, in the \PQ\ languages, to speak of the semantics of a term \textit{in isolation}: a term is always evaluated in the context of an underlying circuit that supplies all of the term's free labels. We therefore define the operational semantics of \PQR\ as a big-step evaluation relation $\eval$ on \emph{configurations}, i.e. circuits paired with either terms or values. Intuitively, $\config{\circuitOne}{\termOne}\eval\config{\circuitTwo}{\valOne}$ means that $\termOne$ evaluates to $\valOne$ and updates $\circuitOne$ to $\circuitTwo$ as a side effect.

\begin{figure}[!ht]
	\centering
	\fbox{\begin{mathpar}
			\inference[\textit{app}]
			{\config{\circuitOne}{\termOne\sub{\valOne}{\varOne}}
				\eval \config{\circuitTwo}{\valTwo}}
			{\config{\circuitOne}{\app{(\abs{\varOne}{\typeOne}{\termOne})}{\valOne}}
				\eval \config{\circuitTwo}{\valTwo}}
			\and
			\inference[\textit{dest}]
			{\config{\circuitOne}{\termOne\sub{\valOne}{\varOne}\sub{\valTwo}{\varTwo}}
				\eval \config{\circuitTwo}{\valThree}}
			{\config{\circuitOne}{\dest{\varOne}{\varTwo}{\tuple{\valOne}{\valTwo}}{\termOne}}
				\eval \config{\circuitTwo}{\valThree}}
			\and
			\inference[\textit{force}]
			{\config{\circuitOne}{\termOne} \eval \config{\circuitTwo}{\valOne}}
			{\config{\circuitOne}{\force{(\lift{\termOne})}} \eval \config{\circuitTwo}{\valOne}}
			\and
			\inference[\textit{apply}]
			{\config{\circuitThree}{\struct{\labFour}} = \append{\circuitOne}{\struct{\labThree}}{\boxedCirc{\struct\labOne}{\circuitTwo}{\struct{\labTwo}}}}
			{\config{\circuitOne}{\apply{\boxedCirc{\struct\labOne}{\circuitTwo}{\struct{\labTwo}}}{\struct{\labThree}}} \eval \config{\circuitThree}{\struct{\labFour}}}
			\and
			\inference[\textit{box}]
			{(\lcOne,\struct{\labOne})=\freshlabels{\mtypeOne}
				&
				\config{\cidentity{\lcOne}}{\termOne} \eval \config{\cidentity{\lcOne}}{\valOne}
				&
				\config{\cidentity{\lcOne}}{\app{\valOne}{\struct{\labOne}}} \eval \config{\circuitTwo}{\struct{\labTwo}}}
			{\config{\circuitOne}{\boxt{\mtypeOne}{(\lift{\termOne})}} \eval \config{\circuitOne}{\boxedCirc{\struct{\labOne}}{\circuitTwo}{\struct{\labTwo}}}}
			\and
			\inference[\textit{return}]
			{ }
			{\config{\circuitOne}{\return{\valOne}} \eval \config{\circuitOne}{\valOne}}
			\and
			\inference[\textit{let}]
			{\config{\circuitOne}{\termOne} \eval \config{\circuitThree}{\valOne}
				&
				\config{\circuitThree}{\termTwo\sub{\valOne}{\varOne}} \eval \config{\circuitTwo}{\valTwo}}
			{\config{\circuitOne}{\letin{\varOne}{\termOne}{\termTwo}} \eval \config{\circuitTwo}{\valTwo}}
			\and
			\inference[\textit{fold-end}]
			{ }
			{\config{\circuitOne}{\app{(\fold{\ivarOne}{\valOne}{\valTwo})}{\nil}} \eval \config{\circuitOne}{\valTwo}}
			\and
			\inference[\textit{fold-step}]
			{\config{\circuitOne}{\termOne\isub{0}{\ivarOne}} \eval \config{\circuitOne}{\valFour}
				&
				\config{\circuitOne}{\app{\valFour}{\tuple{\valOne}{\valTwo}}} \eval \config{\circuitThree}{\valFive}
				\\
				\config{\circuitThree}{\app{(\fold{\ivarOne}{(\lift{\termOne\isub{\iplus{\ivarOne}{1}}{\ivarOne}})}{\valFive})}{\valTwo'}} \eval \config{\circuitTwo}{\valThree}}
			{\config{\circuitOne}{\app{(\fold{\ivarOne}{(\lift{\termOne})}{\valOne})}{(\cons{\valTwo}{\valTwo'})}} \eval \config{\circuitTwo}{\valThree}}
	\end{mathpar}}
	\caption{\PQR\ big-step operational semantics.}
	\label{fig:operational-semantics}
\end{figure}

The rules for evaluating configurations are given in Figure \ref{fig:operational-semantics}, where $\circuitOne,\circuitTwo$ and $\circuitThree$ are circuits, $\termOne$ and $\termTwo$ are terms, while $\valOne,\valTwo,\valThree,\valFour$ and $\valFive$ are values. Most evaluation rules are straightforward, with the exception perhaps of \textit{apply, box} and \textit{fold-step}.
Being the fundamental block of circuit-building, the semantics of $\applyoperator$ lies almost entirely in the way it updates the underlying circuit:

\begin{equation*}
\inference[\textit{apply}]
{\config{\circuitThree}{\struct{\labFour}} = \append{\circuitOne}{\struct{\labThree}}{\boxedCirc{\struct\labOne}{\circuitTwo}{\struct{\labTwo}}}}
{\config{\circuitOne}{\apply{\boxedCirc{\struct\labOne}{\circuitTwo}{\struct{\labTwo}}}{\struct{\labThree}}} \eval \config{\circuitThree}{\struct{\labFour}}}
\end{equation*}

The concatenation of the underlying circuit $\circuitOne$ and the applicand $\circuitTwo$ is delegated entirely to the $\appendfunction$ function, which is defined as follows.

\begin{definition}[$\appendfunction$]
	We define \emph{the append of $\boxedCirc{\struct\labOne}{\circuitTwo}{\struct{\labTwo}}$ to $\circuitOne$ on $\struct\labThree$}, written $\append{\circuitOne}{\struct\labThree}{\boxedCirc{\struct\labOne}{\circuitTwo}{\struct{\labTwo}}}$, as the function that performs the following steps:
	\begin{enumerate}
		\item Finds $\boxedCirc{\struct\labThree}{\circuitTwo'}{\struct{\labFour}}$ equivalent to $\boxedCirc{\struct\labOne}{\circuitTwo}{\struct{\labTwo}}$ such that the labels shared by $\circuitOne$ and $\circuitTwo'$ are all and only those in $\struct\labThree$,
		\item Computes $\circuitThree = \concat{\circuitOne}{\circuitTwo'}$,
		\item Returns $\config{\circuitThree}{\struct{\labFour}}$.
	\end{enumerate}
\end{definition}

Note that two circuits are \emph{equivalent} when they only differ by a renaming of labels, that is, when they have the same fundamental structure. What the renaming does, in this case, is instantiate the generic input interface $\struct\labOne$ of circuit $\circuitTwo$ with the actual labels that it is going to be appended to, namely $\struct\labThree$, and ensure that there are no name clashes between the labels occurring in the resulting $\circuitTwo'$ and those occurring in $\circuitOne$.

On the other hand, the semantics of a term of the form $\boxt{\mtypeOne}{(\lift{\termOne})}$ relies on the $\freshlabelsfunction$ function:

\begin{equation*}
\inference[\textit{box}]
{(\lcOne,\struct{\labOne})=\freshlabels{\mtypeOne}
	&
	\config{\cidentity{\lcOne}}{\termOne} \eval \config{\cidentity{\lcOne}}{\valOne}
	&
	\config{\cidentity{\lcOne}}{\app{\valOne}{\struct{\labOne}}} \eval \config{\circuitTwo}{\struct{\labTwo}}}
{\config{\circuitOne}{\boxt{\mtypeOne}{(\lift{\termOne})}} \eval \config{\circuitOne}{\boxedCirc{\struct{\labOne}}{\circuitTwo}{\struct{\labTwo}}}}
\end{equation*}

What $\freshlabelsfunction$ does is take as input a bundle type $\mtypeOne$ and instantiate fresh $\lcOne,\struct\labOne$ such that $\mjudgment{\lcOne}{\struct\labOne}{\mtypeOne}$. The wire bundle $\struct\labOne$ is then used as a dummy argument to $\valOne$, the circuit-building function resulting from the evaluation of $\termOne$. This function application is evaluated in the context of the identity circuit $\cidentity{\lcOne}$ and eventually produces a circuit $\circuitTwo$, together with its output labels $\struct{\labTwo}$. Finally, $\struct{\labOne}$ and $\struct{\labTwo}$ become respectively the input and output interfaces of the resulting boxed circuit $\boxedCirc{\struct{\labOne}}{\circuitTwo}{\struct{\labTwo}}$.
Note, at this point, that $\mtypeOne$ controls how many labels are initialized by the $\freshlabelsfunction$ function. Because $\mtypeOne$ can contain indices (e.g. it could be that $\mtypeOne \equiv \listt{3}{\qubitt}$), it follows that in \PQR\ indices are not only relevant to typing, but they also have operational value. For this reason, the semantics of \PQR\ is well-defined only on terms closed both in the sense of regular variables \textit{and} index variables, since a circuit-building function of input type, e.g., $\listt{\ivarOne}{\qubitt}$ does not correspond to any individual circuit, and therefore it makes no sense to try and box it.

The operational significance of indices is also apparent in the \textit{fold-step} rule:

\begin{equation*}
\inference[\textit{fold-step}]
{\config{\circuitOne}{\termOne\isub{0}{\ivarOne}} \eval \config{\circuitOne}{\valFour}
	&
	\config{\circuitOne}{\app{\valFour}{\tuple{\valOne}{\valTwo}}} \eval \config{\circuitThree}{\valFive}
	\\
	\config{\circuitThree}{\app{(\fold{\ivarOne}{(\lift{\termOne\isub{\iplus{\ivarOne}{1}}{\ivarOne}})}{\valFive})}{\valTwo'}} \eval \config{\circuitTwo}{\valThree}}
{\config{\circuitOne}{\app{(\fold{\ivarOne}{(\lift{\termOne})}{\valOne})}{(\cons{\valTwo}{\valTwo'})}} \eval \config{\circuitTwo}{\valThree}}
\end{equation*}

Here, the index variable $\ivarOne$ occurring free in $\termOne$ is instantiated to $0$ before evaluating $\termOne$ to obtain the step function $\valFour$. Next, the new accumulator $\valFive$ is computed. Then, before evaluating the next iteration, $\ivarOne$ is replaced with $\iplus{\ivarOne}{1}$ in $\termOne$. This way, each time $\termOne$ is evaluated, $\ivarOne$ is equal to the number of the current iteration, and the evaluation can result in a function $\valFour$ which is operationally distinct for each iteration.

%% file: type-safety-and-correctness.tex
\label{sec:type-safety-and-correctness}

Because the operational semantics of \PQR\ is based on configurations, we ought to adopt a notion of well-typedness which is also based on configurations. The following definition of \textit{well-typed configuration} is thus central to our type-safety analysis. 

\begin{definition}[Well-typed Configuration]
	We say that configuration $\config{\circuitOne}{\termOne}$ is \emph{well-typed with input $\lcOne$, type $\typeOne$, width $\indexOne$ and output $\lcTwo$}, and we write $\cconfigjudgment{\lcOne}{\config{\circuitOne}{\termOne}}{\typeOne}{\indexOne}{\lcTwo}$, whenever $\circjudgment{\circuitOne}{\lcOne}{\lcTwo,\lcThree}$ for some $\lcThree$ such that $\cjudgment{\emptycontext}{\emptycontext}{\lcThree}{\termOne}{\typeOne}{\indexOne}$.
	We write $\vconfigjudgment{\lcOne}{\config{\circuitOne}{\valOne}}{\typeOne}{\lcTwo}$ whenever $\circjudgment{\circuitOne}{\lcOne}{\lcTwo,\lcThree}$ for some $\lcThree$ such that $\vjudgment{\emptycontext}{\emptycontext}{\lcThree}{\valOne}{\typeOne}$.
\end{definition}

The three results that we want to show in this section are that any well-typed term configuration $\cconfigjudgment{\lcOne}{\config{\circuitOne}{\termOne}}{\typeOne}{\indexOne}{\lcTwo}$ evaluates to some configuration $\config{\circuitTwo}{\valOne}$, that $\vconfigjudgment{\lcOne}{\config{\circuitTwo}{\valOne}}{\typeOne}{\lcTwo}$ and that $\circuitTwo$ is obtained from $\circuitOne$ by extending it with a sub-circuit of width at most $\indexOne$. These claims correspond to the \textit{subject reduction} and \textit{total correctness} properties that we will prove at the end of this section. However, both these results rely on a central lemma and on the mutual notions of \textit{realization} and \textit{reducibility}, which we first give formally.

\begin{definition}[Realization]
	We define $\realizes{\valOne}{\lcOne}{\typeOne}$, which reads \emph{$\valOne$ realizes $\typeOne$ under $\lcOne$}, as the smallest relation such that
	\begin{align*}
		\realizes{\unitv&}{\emptylc}{\unitt}\\
		\realizes{\labOne&}{\labOne:\wtypeOne}{\wtypeOne}\\
		\realizes{\valOne&}{\lcOne}{\arrowt{\typeOne}{\typeTwo}{\indexOne}{\indexTwo}} \text{ iff } \eqjudgment{}{\indexTwo}{|\lcOne|}\text{ and }\forall \valTwo: \realizes{\valTwo}{\lcTwo}{\typeOne} \implies \reducible{\app{\valOne}{\valTwo}}{\lcOne,\lcTwo}{\indexOne}{\typeTwo}\\
		\realizes{\lift{\termOne}&}{\emptylc}{\bang{\typeOne}}\text{ iff }\reducible{\termOne}{\emptylc}{0}{\typeOne}\\
		\realizes{\tuple{\valOne}{\valTwo}&}{\lcOne,\lcTwo}{\tensor{\typeOne}{\typeTwo}}\text{ iff }\realizes{\valOne}{\lcOne}{\typeOne}\text{ and }\realizes{\valTwo}{\lcTwo}{\typeTwo}\\
		\realizes{\nil&}{\emptycontext}{\listt{\indexOne}{\typeOne}}\text{ iff }\eqjudgment{}{\indexOne}{0}\\
		\realizes{\cons{\valOne}{\valTwo}&}{\lcOne,\lcTwo}{\listt{\indexOne}{\typeOne}}\text{ iff }\eqjudgment{}{\indexOne}{\iplus{\indexTwo}{1}}\text{ and }\realizes{\valOne}{\lcOne}{\typeOne}\text{ and } \realizes{\valTwo}{\lcTwo}{\listt{\indexTwo}{\typeOne}}\\
		\realizes{\boxedCirc{\struct\labOne}{\circuitOne}{\struct{\labTwo}}&}{\emptylc}{\circt{\indexOne}{\mtypeOne}{\mtypeTwo}}\text{ iff }\circjudgment{\circuitOne}{\lcOne}{\lcTwo}\text{ and } \mjudgment{\lcOne}{\struct\labOne}{\mtypeOne}\text{ and } \mjudgment{\lcTwo}{\struct{\labTwo}}{\mtypeTwo}\text{ and }\leqjudgment{}{\width{\circuitOne}}{\indexOne}
	\end{align*}
\end{definition}

\begin{definition}[Reducibility]
	We say that \emph{$\termOne$ is reducible under $\lcOne$ with type $\typeOne$ and width $\indexOne$}, and we write $\reducible{\termOne}{\lcOne}{\indexOne}{\typeOne}$, if, for all $\circuitOne$ such that $\circjudgment{\circuitOne}{\lcTwo}{\lcOne,\lcThree}$, there exist $\circuitTwo,\valOne$ such that
	\begin{enumerate}
		\item $\config{\circuitOne}{\termOne} \eval \config{\concat{\circuitOne}{\circuitTwo}}{\valOne},$
		\item $\leqjudgment{}{\width{\circuitTwo}}{\indexOne}$
		\item $\circjudgment{\circuitTwo}{\lcOne}{\lcFour}$ for some $\lcFour$ such that $\realizes{\valOne}{\lcFour}{\typeOne}$.
	\end{enumerate}
\end{definition}

Both relations, and in particular reducibility, are given in the form of unary logical relations \cite{logical-relations}. The intuition is pretty straightforward: a term is reducible with width $\indexOne$ if it evaluates correctly when paired with any circuit $\circuitOne$ which provides its free labels and if it extends $\circuitOne$ with a sub-circuit $\circuitTwo$ whose width is bounded by $\indexOne$. Realization, on the other hand, is less immediate. For most cases, realizing type $\typeOne$ loosely corresponds to being closed and well-typed with type $\typeOne$, but a value realizes an arrow type $\arrowt{\typeOne}{\typeTwo}{\indexOne}{\indexTwo}$ when its application to a value realizing $\typeOne$ is reducible with type $\typeTwo$ and width $\indexOne$.

By themselves, realization and reducibility are defined only on terms and values closed in the sense both of regular and index variables. To extend these notions to open terms and values, we adopt the standard approach of reasoning explicitly about the substitutions that could render them closed.
\begin{definition}[Closing Substitution]
	We define the set $\vsubSet$ of \emph{closing value substitutions} as the smallest subset of $\valSet\cup\termSet \to \valSet\cup\termSet$ such that
	\begin{itemize}
		\item $\emptysub\in\vsubSet$ with $\emptysub(\termOne) = \termOne$.
		\item If $\vsubOne\in\vsubSet$, $\varOne$ is a variable name and $\valOne\in\valSet$ is closed, then $\vsubOne[\varOne\mapsto\valOne]\in\vsubSet$ with $\vsubOne[\varOne\mapsto\valOne](\termOne) = \vsubOne(\termOne\sub{\valOne}{\varOne})$.
	\end{itemize}
	We define the set $\isubSet$ of \emph{closing index substitutions} as the smallest subset of $\indexSet\cup\typeSet\cup\valSet\cup\termSet \to \indexSet\cup\typeSet\cup\valSet\cup\termSet$ such that
	\begin{itemize}
		\item $\emptysub\in\isubSet$ with $\emptysub(\termOne) = \termOne$.
		\item If $\isubOne\in\isubSet$, $\ivarOne$ is an index variable name and $\indexOne\in\indexSet$ is closed, then $\isubOne[\ivarOne\mapsto\indexOne]\in\isubSet$ with $\isubOne[\ivarOne\mapsto\indexOne](\termOne) = \isubOne(\termOne\isub{\indexOne}{\ivarOne})$.
	\end{itemize}
\end{definition}

We say that $\vsubOne$ \textit{implements} a typing context $\contextOne$ using label context $\lcOne$, and we write $\vimplements{\vsubOne}{\lcOne}{\contextOne}$, when it replaces every variable $\varOne_i$ in the domain of $\contextOne$ with a value $\valOne_i$ such that $\realizes{\valOne_i}{\lcOne_i}{\contextOne(\varOne_i)}$ and $ \lcOne=\biguplus_{\varOne_i\in\operatorname{dom}(\contextOne)} \lcOne_i$. Similarly, we say that $\isubOne$ implements an index context $\icontextOne$, and we write $\iimplements{\isubOne}{\icontextOne}$, when it replaces every index variable in $\icontextOne$ with a closed index term.
This allows us to give the following fundamental lemma, which will be used while proving all other claims.

\begin{lemma}[Core Correctness]\label{lem:core-correctness}
	Let $\Pi$ be a type derivation. For all $\iimplements{\isubOne}{\icontextOne}$ and $\vimplements{\vsubOne}{\lcOne}{\isubOne(\contextOne)}$, we have that
	\begin{align*}
		\Pi \proves \cjudgment{\icontextOne}{\contextOne}{\lcTwo}{\termOne}{\typeOne}{\indexOne} &\implies \reducible{\vsubOne(\isubOne(\termOne))}{\lcOne,\lcTwo}{\isubOne(\indexOne)}{\isubOne(\typeOne)}\\
		\Pi \proves \vjudgment{\icontextOne}{\contextOne}{\lcTwo}{\valOne}{\typeOne} &\implies \realizes{\vsubOne(\isubOne(\valOne))}{\lcOne,\lcTwo}{\isubOne(\typeOne)}
	\end{align*}
\end{lemma}
\begin{proof}
	By induction on the size of $\Pi$, making use of Theorem \ref{thm:crl}.
\end{proof}

Lemma \ref{lem:core-correctness} tells us that any well-typed term (resp. value) is reducible (resp. realizes its type) when we instantiate its free variables according to its contexts.
Now that we have Lemma \ref{lem:core-correctness}, we can proceed to proving the aforementioned results of subject reduction and total correctness. We start with the former, which unsurprisingly requires the following substitution lemmata.

\begin{lemma}[Index Substitution]\label{lem:index-substitution}
	Let $\Pi$ be a type derivation and let $\indexOne$ be an index such that $\wfjudgment{\icontextOne}{\indexOne}$. We have that
	\begin{align*}
		\Pi \proves \cjudgment{\icontextOne,\ivarOne}{\contextOne}{\lcOne}{\termOne}{\typeOne}{\indexTwo} &\implies \cjudgment{\icontextOne}{\contextOne\isub{\indexOne}{\ivarOne}}{\lcOne}{\termOne\isub{\indexOne}{\ivarOne}}{\typeOne\isub{\indexOne}{\ivarOne}}{\indexTwo\isub{\indexOne}{\ivarOne}},\\
		\Pi \proves \vjudgment{\icontextOne,\ivarOne}{\contextOne}{\lcOne}{\valOne}{\typeOne} &\implies \vjudgment{\icontextOne}{\contextOne\isub{\indexOne}{\ivarOne}}{\lcOne}{\valOne\isub{\indexOne}{\ivarOne}}{\typeOne\isub{\indexOne}{\ivarOne}}.
	\end{align*}
\end{lemma}
\begin{proof}
	By induction on the size of $\Pi$.
\end{proof}

\begin{lemma}[Value Substitution]\label{lem:value-substitution}
	Let $\Pi$ be a type derivation and let $\valOne$ be a value such that $\vjudgment{\icontextOne}{\pcontextOne,\contextOne_1}{\lcOne_1}{\valOne}{\typeOne}$. We have that
	\begin{align*}
		\Pi \proves \cjudgment{\icontextOne}{\pcontextOne,\contextOne_2,\varOne:\typeOne}{\lcOne_2}{\termOne}{\typeTwo}{\indexOne} &\implies \cjudgment{\icontextOne}{\pcontextOne,\contextOne_1,\contextOne_2}{\lcOne_1,\lcOne_2}{\termOne\sub{\valOne}{\varOne}}{\typeTwo}{\indexOne},\\
		\Pi \proves \vjudgment{\icontextOne}{\pcontextOne,\contextOne_2,\varOne:\typeOne}{\lcOne_2}{\valTwo}{\typeTwo} &\implies \vjudgment{\icontextOne}{\pcontextOne,\contextOne_1,\contextOne_2}{\lcOne_1,\lcOne_2}{\valTwo\sub{\valOne}{\varOne}}{\typeTwo}.
	\end{align*}
\end{lemma}
\begin{proof}
	By induction on the size of $\Pi$.
\end{proof}

The main result is then stated in the following theorem.

\begin{theorem}[Subject Reduction]
	If $\cconfigjudgment{\lcOne}{\config{\circuitOne}{\termOne}}{\typeOne}{\indexOne}{\lcTwo}$ and $\config{\circuitOne}{\termOne} \eval \config{\circuitTwo}{\valOne}$, then $\vconfigjudgment{\lcOne}{\config{\circuitTwo}{\valOne}}{\typeOne}{\lcTwo}$.
\end{theorem}
\begin{proof}
	By induction on the derivation of $\config{\circuitOne}{\termOne} \eval \config{\circuitTwo}{\valOne}$ and case analysis on the last rule used in its derivation. Lemma \ref{lem:value-substitution} is essential to the \textit{app,dest} and \textit{let} cases, while Lemma \ref{lem:index-substitution} is used in the \textit{fold-step} case. Lemma \ref{lem:core-correctness} is essential to the \textit{box} case, as it is the only case in which the side effect of the evaluation (the circuit built by the function being boxed), whose preservation is the a matter of correctness, becomes a value (the resulting boxed circuit).
\end{proof}

Of course, type soundness is not enough: we also want the resource analysis carried out by our type system to be correct, as stated in the following theorem.

\begin{theorem}[Total Correctness]
	If $\cconfigjudgment{\lcOne}{\config{\circuitOne}{\termOne}}{\typeOne}{\indexOne}{\lcTwo}$, then there exist $\circuitTwo,\valOne$ such that $\config{\circuitOne}{\termOne} \eval \config{\concat{\circuitOne}{\circuitTwo}}{\valOne}$ and $\leqjudgment{}{\width{\circuitTwo}}{\indexOne}$.
\end{theorem}
\begin{proof}
	By definition, $\cconfigjudgment{\lcOne}{\config{\circuitOne}{\termOne}}{\typeOne}{\indexOne}{\lcTwo}$ entails that $\circjudgment{\circuitOne}{\lcOne}{\lcTwo,\lcThree}$ and $\cjudgment{\emptycontext}{\emptycontext}{\lcThree}{\termOne}{\typeOne}{\indexOne}$. Since an empty context is trivially implemented by an empty closing substitution, by Lemma \ref{lem:core-correctness} we get $\reducible{\termOne}{\lcThree}{\indexOne}{\typeOne}$, which by definition entails that there exist $\circuitTwo,\valOne$ such that $\config{\circuitOne}{\termOne}\eval\config{\concat{\circuitOne}{\circuitTwo}}{\valOne}$ and $\leqjudgment{}{\width{\circuitTwo}}{\indexOne}$.
\end{proof}

%% file: practical-examples.tex
\label{sec:practical-examples}

This section provides an example of how \PQR\ can be used to verify the resource usage of realistic quantum algorithms. In particular, we use our language to implement the QFT algorithm \cite{qft,nielsen-chuang} and verify that the circuits it produces have width no greater than the size of their input, i.e. that the QFT algorithm does not overall use additional ancillary qubits.

\begin{figure}[!ht]
	\centering
	\fbox{\parbox{.98\textwidth}{\centering
		\begin{align*}
			\mathit{qft} &\triangleq \fold{\ivarTwo}{\mathit{qftStep}}{\nil}\\
			\mathit{qftStep} &\triangleq \lift{(\return{\abs{\tuple{\mathit{qs}}{q}}{\tensor{\listt{\ivarTwo}{\qubitt}}{\qubitt}}{\\&
						\letin{\tuple{n}{\mathit{qs}}}{\app{\mathit{qlen}}{\mathit{qs}}}{\\&
							\letin{\mathit{revQs}}{\app{\mathit{rev}}{\mathit{qs}}}{\\&
								\letin{\tuple{q}{\mathit{qs}}}{\app{(\fold{\ivarThree}{(\lift{(\app{\mathit{rotate}}{n})})}{\tuple{q}{\nil}})}{\mathit{revQs}}}{\\&
									\letin{q}{\apply{\mathsf{H}}{q}}{\\&
										\return{(\cons{q}{\mathit{qs}})}}}}}}}})
			\\\\
			\mathit{rotate} &\triangleq \abs{n}{\natt}{\return{\abs{\tuple{\tuple{q}{\mathit{cs}}}{c}}{\tensor{(\tensor{\qubitt}{\listt{\ivarThree}{\qubitt}})}{\qubitt}}{\\&
						\letin{\tuple{m}{\mathit{cs}}}{\app{\mathit{qlen}}{\mathit{cs}}}{\\&
							\letin{\mathit{rgate}}{\app{\mathsf{makeRGate}}{(n+1-m)}}{\\&
								\letin{\tuple{q}{c}}{\apply{\mathit{rgate}}{\tuple{q}{c}}}{\\&
									\return{\tuple{q}{\cons{c}{\mathit{cs}}}}}}}}}}
									\end{align*}
	}}
	\caption{A \PQR\ implementation of the Quantum Fourier Transform circuit family. The usual syntactic sugar is employed.}
	\label{fig:qft-implementation}
\end{figure}

The \PQR\ implementation of the QFT algorithm is given in Figure \ref{fig:qft-implementation}. As we walk through the various parts of the program, be aware that we will focus on the resource aspects of the algorithm, ignoring much of its actual meaning. Starting bottom-up, we assume that we have an encoding of naturals in the language and that we can perform arithmetic on them. We also assume some primitive gates and gate families: $\mathsf{H}$ is the boxed circuit corresponding to the Hadamard gate and has type $\circt{1}{\qubitt}{\qubitt}$, whereas the $\mathsf{makeRGate}$ function has type $\arrowt{\natt}{\circt{2}{\tensor{\qubitt}{\qubitt}}{\tensor{\qubitt}{\qubitt}}}{0}{0}$ and produces instances of the parametric controlled $R_n$ gate.

\begin{figure}[!ht]
	\centering
	\fbox{\parbox{.98\textwidth}{
			\begin{align*}
				\mathit{qlen} &\triangleq \abs{\mathit{ql}}{\listt{\ivarOne}{\qubitt}}{\\&
					\letin{\varOne}{\app{(\fold{\ivarTwo}{\mathit{qlenStep}}{\tuple{0}{\nil}})}{\mathit{ql}}}{\\&
						\dest{n}{\mathit{revQl}}{\varOne}{\\&
							\letin{\mathit{preQl}}{\app{\mathit{rev}}{\mathit{revQl}}}{\\&
								\return{\tuple{n}{\mathit{preQl}}}}}}}
				\\
				\mathit{qlenStep} &\triangleq \lift{(\return{\abs{\varOne}{\tensor{(\tensor{\natt}{\listt{\ivarOne}{\qubitt}})}{\qubitt}}{\\&
							\dest{\mathit{acc}}{q}{\varOne}{\\&
								\dest{n}{\mathit{revQl}}{\mathit{acc}}{\\&
									\return{\tuple{n+1}{\cons{q}{\mathit{revQl}}}}}}}})}
				\\\\
				\mathit{rev} &\triangleq \fold{\ivarOne}{\mathit{revStep}}{\nil}
				\\
				\mathit{revStep} &\triangleq \lift{(\return{\abs{\varOne}{\tensor{\listt{\ivarOne}{\qubitt}}{\qubitt}}{\\&
							\dest{\mathit{revQl}}{q}{\varOne}{\return({\cons{q}{\mathit{revQl}})}}}})}
	\end{align*}}}
	\caption{The implementation of the auxiliary functions \textit{qlen} and \textit{rev}.}
	\label{fig:aux-fun-implementation}
\end{figure}

On the other hand, $\mathit{qlen}$ and $\mathit{rev}$ stand for regular language terms which implement respectively the linear list length and reverse functions. Their implementation is given in Figure \ref{fig:aux-fun-implementation} and they have types $\mathit{qlen}::\arrowt{\listt{\ivarOne}{\qubitt}}{(\tensor{\natt}{\listt{\ivarOne}{\qubitt}})}{\ivarOne}{0}$ and $\mathit{rev}::\arrowt{\listt{\ivarOne}{\qubitt}}{\listt{\ivarOne}{\qubitt}}{\ivarOne}{0}$ in our type system.
We now turn our attention to the actual QFT algorithm. Function $\mathit{qftStep}$ builds a single step of the QFT circuit. The width of the circuit produced at step $\ivarTwo$ is dominated by the folding of the $\app{\mathit{rotate}}{n}$ function, which applies controlled rotations between appropriate pairs of qubits and has type
\begin{equation}
	{\arrowt{\tensor{(\tensor{\qubitt}{\listt{\ivarThree}{\qubitt}})}{\qubitt}}{\tensor{\qubitt}{\listt{\iplus{\ivarThree}{1}}{\qubitt}}}{\iplus{\ivarThree}{2}}{0}},
\end{equation}
meaning that $\app{\mathit{rotate}}{n}$ rearranges the structure of its inputs, but overall does not introduce any new wire. We fold this function starting from an accumulator $\tuple{q}{\nil}$, meaning that we can give $\fold{\ivarTwo}{(\lift{(\app{\mathit{rotate}}{n}}))}{\tuple{q}{\nil}}$ type as follows:
\begin{equation}
	\inference[\textit{fold}]
	{\vjudgment{\ivarOne,\ivarTwo,\ivarThree}{n:\natt}{\emptycontext}{\lift{(\app{\mathit{rotate}}{n}})}{\bang{({\arrowt{\tensor{(\tensor{\qubitts}{\listt{\ivarThree}{\qubitts}})}{\qubitts}}{\tensor{\qubitts}{\listt{\iplus{\ivarThree}{1}}{\qubitts}}}{\iplus{\ivarThree}{2}}{0}})}}
		\\
		\vjudgment{\ivarOne,\ivarTwo}{q:\qubitts}{\emptycontext}{\tuple{q}{\nil}}{\tensor{\qubitts}{\listt{0}{\qubitts}}}
		&
		\wfjudgment{\ivarOne,\ivarTwo}{\ivarTwo}
		&
		\wfjudgment{\ivarOne,\ivarTwo}{\qubitts}
		& 
	}
	{\vjudgment{\ivarOne,\ivarTwo}{n:\natt,q:\qubitts}{\emptycontext}{\fold{\ivarThree}{\lift{(\app{\mathit{rotate}}{n}})}{\tuple{q}{\nil}}}{\arrowt{\listt{\ivarTwo}{\qubitts}}{\tensor{\qubitts}{\listt{\ivarTwo}{\qubitts}}}{\iplus{\ivarTwo}{1}}{1}}}
\end{equation}
where $\qubitts$ is shorthand for $\qubitt$ and where we implicitly use the fact that $\eqjudgment{\ivarOne,\ivarTwo}{\imax{1}{\imaximum{\ivarThree}{\ivarTwo}{\iplus{\iplus{\ivarThree}{2}}{\imult{(\iminus{\iminus{\ivarTwo}{1}}{\ivarThree})}{1}}}}}{\iplus{\ivarTwo}{1}}$ to simplify the arrow's width annotation using \textit{vsub} and the \textit{arrow} subtyping rule. Next, we fold over $\mathit{revQs}$, which has the same elements as $\mathit{qs}$ and thus has length $\ivarTwo$, and we obtain that the fold produces a circuit whose width is bounded by $\iplus{\ivarTwo}{1}$. Therefore, $\mathit{qftStep}$ has type
\begin{equation}
	\bang{(\arrowt{\tensor{(\listt{\ivarTwo}{\qubitt}}{\qubitt})}{\listt{\iplus{\ivarTwo}{1}}{\qubitt}}{\iplus{\ivarTwo}{1}}{0})},
\end{equation}
which entails that when we pass it as an argument to the topmost $\foldoperator$ together with $\nil$ we can conclude that the type of the $\mathit{qft}$ function is

\begin{equation}
	\inference[\textit{fold}]
	{
		\vjudgment{\ivarOne,\ivarTwo}{\emptycontext}{\emptycontext}{\mathit{qftStep}}{\bang{(\arrowt{\tensor{(\listt{\ivarTwo}{\qubitt}}{\qubitt})}{\listt{\iplus{\ivarTwo}{1}}{\qubitt}}{\iplus{\ivarTwo}{1}}{0})}}
		\\
		\vjudgment{\ivarOne}{\emptycontext}{\emptycontext}{\nil}{\listt{0}{\qubitt}}
		&
		\wfjudgment{\ivarOne}{\ivarOne}
		&
		\wfjudgment{\ivarOne}{\qubitt}
	}
	{\vjudgment{\ivarOne}{\emptycontext}{\emptycontext}{\fold{\ivarTwo}{\mathit{qftStep}}{\nil}}{\arrowt{\listt{\ivarOne}{\qubitt}}{\listt{\ivarOne}{\qubitt}}{\ivarOne}{0}}}
\end{equation}
where we once again implicitly simplify the arrow type using the fact that $\eqjudgment{\ivarOne}{\imax{0}{\imaximum{\ivarTwo}{\ivarOne}{\iplus{\iplus{\ivarTwo}{1}}{\imult{(\iminus{\iminus{\ivarOne}{1}}{\ivarTwo})}{1}}}}}{\ivarOne}$. This concludes our analysis and the resulting type tells us that $\mathit{qft}$ produces a circuit of width at most $\ivarOne$ on inputs of size $\ivarOne$, without overall using any additional wires. If we instantiate $\ivarOne$ to $3$, for example, we can apply $\mathit{qft}$ to a list of $3$ qubits to obtain the circuit shown in Figure \ref{fig:qft-circuit}, whose width is exactly $3$.

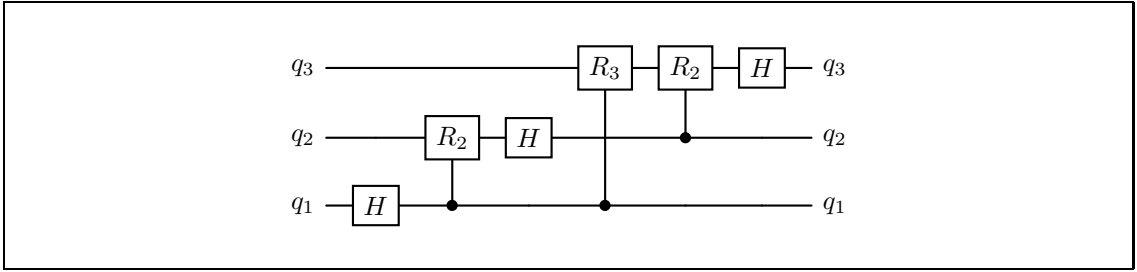
\begin{figure}[!ht]
	\centering
	\fbox{\parbox{.98\textwidth}{\centering
		\begin{quantikz}[column sep=1em, row sep=1em]
			\\
			\lstick{$q_3$} &&&& \gate{R_3} & \gate{R_2} & \gate{H} & \rstick{$q_3$}\\
			\lstick{$q_2$} && \gate{R_2} & \gate{H} && \ctrl{-1} && \rstick{$q_2$}\\
			\lstick{$q_1$} & \gate{H} & \ctrl{-1} && \ctrl{-2} &&& \rstick{$q_1$}\\
			\end{quantikz}
		}}
	\caption{The circuit of input size $3$ produced by $\app{\mathit{qft}}{(\cons{q_1}{\cons{q_2}{\cons{q_3}{\nil}}})}$.}
	\label{fig:qft-circuit}
\end{figure}

To conclude this section, note that for ease of exposition $\mathit{qft}$ actually produces the \textit{reversed} QFT circuit. This is not a problem, since the two circuits are equivalent resource-wise and the actual QFT circuit can be recovered by boxing the result of $\mathit{qft}$ and reversing it via a primitive operator. Besides, note that \Quipper's internal implementation of the QFT is also reversed \cite{quipper-qft}.

%% file: related.tex
The metatheory of quantum circuit description languages, and in 
particular of \Quipper-style languages, has been the subject of quite some 
work in recent years, starting with Ross's thesis on \PQS\ \cite{proto-quipper-s} and going forward with Selinger and Rios's \PQM\ \cite{proto-quipper-m}. In the last five years, some proposals have also appeared for more expressive type systems or for language extensions that can handle non-standard language features, such as the so-called \emph{dynamic lifting} \cite{proto-quipper-l,proto-quipper-dyn,proto-quipper-k}, available in the \Quipper\ language, or dependent types \cite{proto-quipper-d}. Although some embryonic contributions in the direction of analyzing the size of circuits produced using \Quipper\ have been given \cite{gate-count}, no contribution tackles the problem of deriving resource bounds \emph{parametric} on the size of the input. In this, the ability to have types which depend on the input, certainly a feature of \PQD\ \cite{proto-quipper-d}, is not useful for the analysis of intensional attributes of the underlying circuit, simply because such attributes are not visible in types.

If we broaden the horizon to quantum programming languages other than \Quipper, 
we come across, for example, the recent works of Avanzini et al. \cite{quantum-transformer} and Liu et al. \cite{quantum-weakest} on adapting the classic weakest precondition technique to the cost analysis of quantum programs, which however focus on programs in an imperative language. The work of Dal Lago et al. \cite{quantum-implicit-complexity} on a quantum language which characterizes complexity classes for quantum polynomial time should certainly be remembered: even though the language allows the use of higher-order functions, the manipulation of quantum data occurs directly and not through circuits. Similar considerations hold for the recent work of Hainry et al. \cite{quantum-polynomial-time} and Yamakami's algebra of functions \cite{algebra-of-functions} in the style of Bellantoni and Cook \cite{polytime-functions-cook}, both characterizing quantum polynomial time.

If we broaden our scope further and become interested in the analysis of the 
cost of classical or probabilistic programs, we face a 
vast literature, with contributions employing a variety of techniques on heterogeneous languages and calculi: from functional programs \cite{ho-meets-fo,amortized-analysis,resource-aware-ml} and term 
rewriting systems \cite{complexity-rewriting,tyrolean-complexity,giesl-rewriting} to probabilistic \cite{runtime-analysis-probabilistic} and object-oriented programs \cite{java-analysis,aprove}. 
In this context, the resource under analysis is often assumed to be computation 
\emph{time}, which is relatively easy to analyze given its strictly monotonic nature. Circuit width, although monotonically non-decreasing, evolves in a way that depends on a non-monotonic quantity, i.e. the number of wires discarded by a circuit. As a result, width has the flavor of space and its analysis is less straightforward.

It is also worth mentioning that linear dependent types can be seen as a specialized version of refinement types \cite{refinement-ml}, which have been used extensively in the literature to automatically verify interesting properties of programs \cite{refinement-incremental-complexity,refinement-effects}. In particular, the work of Vazou et al. on \LiquidHaskell\ \cite{refinement-haskell,liquid-haskell} has been of particular inspiration, on account of \Quipper\ being embedded precisely in \Haskell. The liquid type system \cite{liquid-types} of \LiquidHaskell\ relies on SMT solvers to discharge proof obligations and has been used fruitfully to reason about both the correctness and the resource consumption (mainly time complexity) of concrete, practical programs \cite{liquidate-your-assets}.

%% file: generalization.tex
\label{sec:generalization-to-other-resource-types}

This work focuses on estimating the \emph{width} of the circuits produced by \Quipper\ programs. This choice is dictated by the fact that the width of a circuit corresponds to the maximum number of distinct wires, and therefore individual qubits, required to execute it. Nowadays, this is considered as one of the most precious resources in quantum computing, and as such must be kept under control. However, this does not mean that our system could not be adapted to the estimation of other parameters. This section outlines how this may be possible.

First, estimating strictly monotonic resources, such as the 
total \emph{number of gates} in a circuit, is possible and in fact simpler 
than estimating width. A \emph{single} index 
term $\indexOne$ that measures the number of gates in the circuit built by a computation would be enough to carry out this analysis. This index would be appropriately increased any time an $\applyoperator$ instruction is executed, while sequencing two terms via $\letoperator$ would simply add together the respective indices.

If the parameter of interest were instead the \emph{depth} of the circuit, then the approach would have to be slightly different. Although in principle it would be possible to still rely on a single index $\indexOne$, this would give rise to a very coarse approximation, effectively collapsing the analysis of depth to a gate count analysis. A more precise approximation could instead be obtained by keeping track of depth \emph{locally} rather than \emph{globally}. More specifically, it would be sufficient to decorate each occurrence of a wire type $\wtypeOne$ with an index term $\indexOne$ so that if a label $\labOne$ were typed with $\wtypeOne^\indexOne$, it would mean that the sub-circuit rooted in $\labOne$ has a depth at most equal to $\indexOne$.

Finally, it should be mentioned that the resources considered, i.e. the depth, width, and gate count of a circuit, can be further refined so as to 
take into account only \emph{some} kinds of wires and gates. For instance, one could want to keep track of the maximum number of \emph{qubits} needed, ignoring the number of classical bits, or at least distinguishing the two parameters, which of course have distinct levels of criticality in current quantum hardware.

%% file: conclusion.tex
In this paper we introduced a linear dependent type system based on index refinements and effect typing for the paradigmatic calculus \PQ, with the purpose of using it to derive upper bounds on the width of the circuits produced by programs. We proved not only the classic type safety properties, but also that the upper bounds derived via the system are correct. We also showed how our system can verify a realistic quantum algorithm and elaborated on some ideas on how our technique could be adapted to other crucial resources types, like gate count and circuit depth.
Ours is the first type system designed specifically for the purpose of resource analysis to target circuit description languages such as \Quipper. Technically, the main novelties are the smooth combination of effect typing and index refinements, but also the proof of correctness, in which reducibility and effects are shown to play well together.

Among topics for further work, we can identify three main research directions. First and foremost, it would be valuable to investigate the ideas presented in this paper from a more practical perspective, that is, to provide a prototype implementation of the language and, more importantly, of the type-checking procedure. This would require understanding the role that SMT solvers may play in discharging the semantic judgments which we use pervasively in our approach.

Staying instead on the theoretical side of things, on the one hand we have the prospect of denotational semantics: most incarnations of \PQ\ are endowed with categorical semantics that model both circuits and the terms of the language that build them \cite{proto-quipper-m,proto-quipper-l,proto-quipper-d,proto-quipper-dyn}. We already mentioned how the intensional nature of the quantity under analysis renders the formulation of an abstract categorical semantics for \PQR\ and its circuits a nontrivial task, but we believe that one such semantics would help \PQR\ fit better in the \PQ\ landscape.

On the other hand, in Section \ref{sec:generalization-to-other-resource-types} we briefly discussed how our system could be modified to handle the analysis of different resource types. It would be interesting to test this path and to investigate the possibility of \textit{actually generalizing} our resource analysis, that is, of making it parametric on the kind of resource being analyzed. This would allow for the same program in the same language to be amenable to different forms of verification, in a very flexible fashion.

%% file: main.bbl
\begin{thebibliography}{10}

\bibitem{qml}
Thorsten Altenkirch and Jonathan Grattage.
\newblock A functional quantum programming language.
\newblock In {\em Proc. of LICS '05}, 2005.

\bibitem{ho-meets-fo}
Martin Avanzini, Ugo Dal~Lago, and Georg Moser.
\newblock Analysing the complexity of functional programs: Higher-order meets
  first-order.
\newblock In {\em Proc. of ICFP 2015}, 2015.

\bibitem{complexity-rewriting}
Martin Avanzini and Georg Moser.
\newblock Complexity analysis by rewriting.
\newblock In {\em Proc. of FLOPS 2008}, 2008.

\bibitem{tyrolean-complexity}
Martin Avanzini and Georg Moser.
\newblock {Tyrolean Complexity Tool: Features and Usage}.
\newblock In {\em Proc. of RTA 2013}, volume~21, 2013.

\bibitem{quantum-transformer}
Martin Avanzini, Georg Moser, Romain Pechoux, Simon Perdrix, and Vladimir
  Zamdzhiev.
\newblock Quantum expectation transformers for cost analysis.
\newblock In {\em Proc. of LICS '22}, 2022.

\bibitem{polytime-functions-cook}
Stephen Bellantoni and Stephen Cook.
\newblock A new recursion-theoretic characterization of the polytime functions
  (extended abstract).
\newblock In {\em STOC '92}, 1992.

\bibitem{smt-solvers}
Armin Biere, Marijn Heule, Hans van Maaren, and Toby Walsh.
\newblock {\em Handbook of Satisfiability - Second Edition}, volume 336 of {\em
  Frontiers in Artificial Intelligence and Applications}.
\newblock {IOS} Press, 2021.

\bibitem{deutsch-jozsa}
R.~Cleve, A.~Ekert, C.~Macchiavello, and M.~Mosca.
\newblock Quantum algorithms revisited.
\newblock {\em Proc. Math. Phys. Eng. Sci. P ROY SOC A-MATH PHY}, 454(1969),
  1998.

\bibitem{proto-quipper-k}
Andrea Colledan and Ugo Dal~Lago.
\newblock {On Dynamic Lifting and Effect Typing in Circuit Description
  Languages}.
\newblock In {\em TYPES 2022}, volume 269, 2023.

\bibitem{osprey}
Hugh Collins and Chris Nay.
\newblock Ibm unveils 400 qubit-plus quantum processor and next-generation ibm
  quantum system two, 2022.
\newblock retrieved on Oct. 15, 2023.

\bibitem{jiuzhang}
Emily Conover.
\newblock Light-based quantum computer jiuzhang achieves quantum supremacy,
  2020.
\newblock retrieved on Oct. 15, 2023.

\bibitem{qft}
D.~Coppersmith.
\newblock An approximate fourier transform useful in quantum factoring, ibm
  research report rc19642, 2002.

\bibitem{linear-dependent-types-relative-completeness}
Ugo Dal~Lago and Marco Gaboardi.
\newblock Linear dependent types and relative completeness.
\newblock In {\em Proc. of LICS '11}, 2011.

\bibitem{quantum-implicit-complexity}
Ugo {Dal Lago}, Andrea Masini, and Margherita Zorzi.
\newblock Quantum implicit computational complexity.
\newblock {\em TCS}, 411(2), 2010.

\bibitem{linear-dependent-types-cbv}
Ugo Dal~lago and Barbara Petit.
\newblock Linear dependent types in a call-by-value scenario.
\newblock In {\em Proc. of PPDP '12}, 2012.

\bibitem{geometry-of-types}
Ugo Dal~lago and Barbara Petit.
\newblock The geometry of types.
\newblock In {\em Proc. of POPL '13}, 2013.

\bibitem{quipper-qft}
Richard Eisenberg, Alexander Green, Peter Lumsdaine, Keith Kim, Siun-Chuon Mau,
  Baranidharan Mohan, Won Ng, Joel Ravelomanantsoa-Ratsimihah, Neil Ross, Artur
  Scherer, Peter Selinger, Benoît Valiron, Alexandr Virodov, and Stephan
  Zdancewic.
\newblock Quipper.libraries.qft.
\newblock Retrieved on Oct. 15, 2023.

\bibitem{quantum-software}
Mark Fingerhuth, Tom{\'{a}}{\v{s}} Babej, and Peter Wittek.
\newblock Open source software in quantum computing.
\newblock {\em {PLOS} {ONE}}, 13(12), 2018.

\bibitem{refinement-ml}
Tim Freeman and Frank Pfenning.
\newblock Refinement types for ml.
\newblock In {\em Proc. of PLDI '91}, 1991.

\bibitem{aprove}
Florian Frohn and J{\"u}rgen Giesl.
\newblock Complexity analysis for java with aprove.
\newblock In {\em Proc. of IFM 2017}, 2017.

\bibitem{proto-quipper-d-intro}
Peng Fu, Kohei Kishida, Neil~J. Ross, and Peter Selinger.
\newblock A tutorial introduction to quantum circuit programming in dependently
  typed proto-quipper.
\newblock In {\em Proc. of RC}, 2020.

\bibitem{proto-quipper-dyn}
Peng Fu, Kohei Kishida, Neil~J. Ross, and Peter Selinger.
\newblock Proto-quipper with dynamic lifting.
\newblock In {\em Proc. of POPL '23}, 2023.

\bibitem{proto-quipper-d}
Peng Fu, Kohei Kishida, and Peter Selinger.
\newblock Linear dependent type theory for quantum programming languages:
  Extended abstract.
\newblock In {\em Proc. of LICS '20}, 2020.

\bibitem{linear-dependent-types-privacy}
Marco Gaboardi, Andreas Haeberlen, Justin Hsu, Arjun Narayan, and Benjamin~C.
  Pierce.
\newblock Linear dependent types for differential privacy.
\newblock In {\em Proc. of POPL '13}, 2013.

\bibitem{survey-gay}
Simon~J. Gay.
\newblock Quantum programming languages: Survey and bibliography.
\newblock {\em Math. Struct. Comput. Sci.}, 16(4), 2006.

\bibitem{quipper-intro}
Alexander~S. Green, Peter~LeFanu Lumsdaine, Neil~J. Ross, Peter Selinger, and
  Beno{\^i}t Valiron.
\newblock An introduction to quantum programming in quipper.
\newblock In {\em Proc. of RC}, 2013.

\bibitem{quipper}
Alexander~S. Green, Peter~LeFanu Lumsdaine, Neil~J. Ross, Peter Selinger, and
  Benoît Valiron.
\newblock Quipper.
\newblock In {\em Proc. of PLDI}, 2013.

\bibitem{grover}
L.~K. Grover.
\newblock A fast quantum mechanical algorithm for database search.
\newblock In {\em Proc. of STOC}, 1996.

\bibitem{java-analysis}
Emmanuel Hainry and Romain P{\'e}choux.
\newblock {Type-based heap and stack space analysis in Java}, 2013.
\newblock Technical report.

\bibitem{quantum-polynomial-time}
Emmanuel Hainry, Romain P\'{e}choux, and M\'{a}rio Silva.
\newblock A programming language characterizing quantum polynomial time.
\newblock In {\em Proc. of FoSSaCS 2023}, 2023.

\bibitem{liquidate-your-assets}
Martin Handley, Niki Vazou, and Graham Hutton.
\newblock Liquidate your assets: reasoning about resource usage in liquid
  haskell.
\newblock {\em PACMPL}, 4, 2019.

\bibitem{hhl}
Aram~W. Harrow, Avinatan Hassidim, and Seth Lloyd.
\newblock Quantum algorithm for linear systems of equations.
\newblock {\em Phys. Rev. Lett.}, 103, 2009.

\bibitem{resource-aware-ml}
Jan Hoffmann, Klaus Aehlig, and Martin Hofmann.
\newblock Resource aware ml.
\newblock In {\em Computer Aided Verification}, 2012.

\bibitem{amortized-analysis}
Jan Hoffmann and Martin Hofmann.
\newblock Amortized resource analysis with polynomial potential.
\newblock In {\em Programming Languages and Systems}, 2010.

\bibitem{runtime-analysis-probabilistic}
Benjamin~Lucien Kaminski, Joost-Pieter Katoen, Christoph Matheja, and Federico
  Olmedo.
\newblock Weakest precondition reasoning for expected run--times of
  probabilistic programs.
\newblock In {\em Programming Languages and Systems}, Berlin, Heidelberg, 2016.

\bibitem{proto-quipper-l}
Dongho Lee, Valentin Perrelle, Beno\^{i}t Valiron, and Zhaowei Xu.
\newblock {Concrete Categorical Model of a Quantum Circuit Description Language
  with Measurement}.
\newblock In {\em Proc. of FSTTCS}, 2021.

\bibitem{quantum-weakest}
Junyi Liu, Li~Zhou, Gilles Barthe, and Mingsheng Ying.
\newblock Quantum weakest preconditions for reasoning about expected runtimes
  of quantum programs.
\newblock In {\em Proc. of LICS '22}, 2022.

\bibitem{refinement-effects}
Yitzhak Mandelbaum, David Walker, and Robert Harper.
\newblock An effective theory of type refinements.
\newblock In {\em Proc. of {ICFP} '03}, 2003.

\bibitem{sycamore}
John Martinis.
\newblock Quantum supremacy using a programmable superconducting processor,
  2019.
\newblock Retrieved on Oct. 15, 2023.

\bibitem{nielsen-chuang}
Michael~A. Nielsen and Isaac~L. Chuang.
\newblock {\em Quantum Computation and Quantum Information: 10th Anniversary
  Edition}.
\newblock Cambridge University Press, 2010.

\bibitem{types-and-effects}
Flemming Nielson and Hanne~Riis Nielson.
\newblock Type and effect systems.
\newblock In {\em Correct System Design: Recent Insights and Advances}. 1999.

\bibitem{giesl-rewriting}
Lars Noschinski, Fabian Emmes, and J\"{u}rgen Giesl.
\newblock Analyzing innermost runtime complexity of term rewriting by
  dependency pairs.
\newblock {\em Journal of Automated Reasoning}, 51(1), 2013.

\bibitem{quantum-languages}
Jens Palsberg.
\newblock Toward a universal quantum programming language.
\newblock {\em XRDS: Crossroads}, 26(1), 2019.

\bibitem{qwire}
Jennifer Paykin, Robert Rand, and Steve Zdancewic.
\newblock Qwire: A core language for quantum circuits.
\newblock In {\em Proc. of POPL '17}, 2017.

\bibitem{quantum-supremacy}
John Preskill.
\newblock Quantum computing and the entanglement frontier, 2012.

\bibitem{proto-quipper-m}
Francisco Rios and Peter Selinger.
\newblock A categorical model for a quantum circuit description language.
\newblock In {\em Proc. of QPL '17}, 2017.

\bibitem{liquid-types}
Patrick~M. Rondon, Ming Kawaguci, and Ranjit Jhala.
\newblock Liquid types.
\newblock In {\em Proc. of PLDI '08}, 2008.

\bibitem{proto-quipper-s}
Neil Ross.
\newblock {\em Algebraic and Logical Methods in Quantum Computation}.
\newblock PhD thesis, 2015.

\bibitem{qgcl}
J.~W. Sanders and P.~Zuliani.
\newblock Quantum programming.
\newblock In {\em Proc. of MPC 2000}, 2000.

\bibitem{quantum-decoherence}
Maximilian Schlosshauer.
\newblock {\em Decoherence and the Quantum-To-Classical Transition}.
\newblock Springer Berlin Heidelberg, 2007.

\bibitem{survey-selinger}
Peter Selinger.
\newblock A brief survey of quantum programming languages.
\newblock In {\em Proc. of FLOPS 2004}, 2004.

\bibitem{qpl}
Peter Selinger.
\newblock Towards a quantum programming language.
\newblock {\em Math. Struct. Comput. Sci.}, 14(4), 2004.

\bibitem{quantum-lambda-calculus}
Peter Selinger and Benoît Valiron.
\newblock A lambda calculus for quantum computation with classical control.
\newblock In {\em Proc. of TLCA}, 2005.

\bibitem{shor}
P.W. Shor.
\newblock Algorithms for quantum computation: discrete logarithms and
  factoring.
\newblock In {\em Proc. of FOCS '94}, 1994.

\bibitem{logical-relations}
Lau Skorstengaard.
\newblock An introduction to logical relations, 2019.

\bibitem{gate-count}
Benoit Valiron.
\newblock Automated, parametric gate count of quantum programs, 2016.
\newblock Retrieved on Oct. 15 2023.

\bibitem{liquid-haskell}
Niki Vazou, Eric~L. Seidel, and Ranjit Jhala.
\newblock Liquidhaskell: Experience with refinement types in the real world.
\newblock In {\em Proc. of Haskell '14}, 2014.

\bibitem{refinement-haskell}
Niki Vazou, Eric~L. Seidel, Ranjit Jhala, Dimitrios Vytiniotis, and Simon
  Peyton-Jones.
\newblock Refinement types for {Haskell}.
\newblock In {\em Proc. of ICFP '14}, 2014.

\bibitem{algebra-of-functions}
Tomoyuki Yamakami.
\newblock A schematic definition of quantum polynomial time computability.
\newblock {\em The Journal of Symbolic Logic}, 85(4), 2020.

\bibitem{refinement-incremental-complexity}
Ezgi Çiçek, Deepak Garg, and Umut Acar.
\newblock Refinement {Types} for {Incremental} {Computational} {Complexity}.
\newblock In {\em Programming {Languages} and {Systems}}, volume 9032. 2015.

\bibitem{qcl}
Bernhard Ömer.
\newblock Classical concepts in quantum programming.
\newblock {\em Int. J. Theor. Phys.}, 44(7), 2005.

\end{thebibliography}
